\documentclass[letterpaper, 10 pt, conference]{ieeeconf}

\IEEEoverridecommandlockouts                              % This command is only needed if 
\usepackage{amsmath} 
\usepackage{amssymb}  
\usepackage{mathbbol}
\usepackage{cite}
\usepackage{wrapfig}
\usepackage{mathrsfs}
\usepackage{url}
\usepackage{array}
\usepackage{amsmath,amsfonts, amssymb,color}
\usepackage{graphicx}
\usepackage{tabularx}
\usepackage{upgreek}
\usepackage[OT1]{fontenc}

\newtheorem{theorem}{Theorem}
\newtheorem{corollary}{Corollary}[theorem]
\newtheorem{lemma}{Lemma}
\newtheorem{definition}{Definition}

\newtheorem{proposition}{Proposition}

\newtheorem{remark}{Remark}

\usepackage{graphics}
\usepackage{graphicx}
\usepackage{tabularx}
\usepackage{array}
\usepackage{epstopdf}
\usepackage{lipsum}
\usepackage{tablefootnote}
\usepackage{subfigure}

\newcommand{\zt}{\textbf{z}}
\newcommand{\St}{\texttt{S}}
\newcommand{\It}{\texttt{I}}
\newcommand{\ct}{\mathcal{C}}
\newcommand{\dt}{\mathcal{D}}
\newcommand{\ut}{\mathcal{U}}

\newcommand{\ignore}[1]{}

\title{\LARGE \bf
SIS Epidemic Containment under Game-theoretic Rational Learning
\author{Urmee Maitra$^{1}$, Keith Paarporn$^{2}$ and Ceyhun Eksin$^{1}$% <-this % stops a space
}
\thanks{$^{1}$U. Maitra and C. Eksin are with the Department of Industrial and Systems Engineering at Texas A \& M University, USA. Email: {\tt\small \{umaitra1309, eksinc\}@tamu.edu.}}%
\thanks{$^{2}$K. Paarporn is with the Department of Computer Science at the University
of Colorado, Colorado Springs, USA. Email: {\tt\small kpaarpor@uccs.edu.}
This work is supported by ONR grant \#N000142612120.}%
}

\allowdisplaybreaks
\begin{document}

\maketitle
\thispagestyle{empty} 

%%%%%%%%%%%%%%%%%%%%%%%%%%%%%%%%%%%%%%%%%%%%%%%%%%%%%%%%%%%%%%%%%%%%%%%%%%%%%%%%
\begin{abstract}
We consider a susceptible-infected-susceptible (SIS) epidemiological model in which individuals dynamically decide whether to follow social distancing or not. Driven by game-theoretic payoffs, their actions influence the spread of infection. Much of the existing work in this area considers individuals that are imitative learners, leading to the well-studied replicator equation. In this paper, we investigate how the dynamical outcomes may fundamentally shift when individuals are instead equipped with logit learning rules, i.e., when individuals possess bounded rationality. We analyze the local asymptotic stability of equilibria that exist in different parameter regimes of the coupled dynamical system. We show that in regimes where the disease-free equilibrium is unstable under imitative learning and sufficiently high-rationality logit learning, sufficiently low rationality can avoid such an undesirable outcome by stabilizing the disease-free equilibrium. We also identify parameter regimes in which an unstable endemic equilibrium stabilizes when rationality falls below a threshold. We numerically illustrate the convergence of infection state and response trajectories to a stable disease-free equilibrium at a lower degree of rationality, whereas the same trajectories otherwise converge to a stable endemic equilibrium as the rationality level rises. We also compare a specific case in which trajectories under replicator dynamics fail to converge, and illustrate convergence when rationality is below an analytically derived threshold.
\end{abstract}

%\begin{IEEEkeywords}
%Game theory, Agents-based systems, Decentralized control
%\end{IEEEkeywords}

%%%%%%%%%%%%%%%%%%%%%%%%%%%%%%%%%%%%%%%%%%%%%%%%%%%%%%%%%%%%%%%%%%%%%%%%%%%%%%%%
\section{Introduction}\label{section:introduction}

Containing epidemic spread is a challenging problem. In order to curb infection spread, authorities announce several health recommendations for the public. Whether individuals comply with these recommendations depends on the infection level, and the collective compliance of the population, in turn determines how the infection spreads. Containment of infection spread is a complicated problem primarily due to two reasons. First, epidemic propagation involves a large interacting population that is difficult to model. Second, population behavior is instrumental in shaping the trajectory of infection spread, leading to a coupled dynamical system. Analyzing the system and characterizing equilibria of this coupled system is non-trivial.

Researchers used notions from population games \cite{sandholm2010population} to investigate the coupled epidemiological-behavioral model \cite{reluga2010game,liu2022herd}. In \cite{kabir2020evolutionary}, \cite{ye2021game}, the authors analyzed the coupled epidemiological system using game theory. The authors in \cite{chang2020game} and \cite{huang2022game} provided comprehensive reviews in this field. Prior works have investigated the effects of non-pharmaceutical interventions \cite{wei2023adoption, maitra2023sis}, including information design \cite{maitra2026discrete}, and vaccinations \cite{hota2019game,hota2023learning} on the spread of pandemics through the lens of game theory, and have proposed solutions to control the spread \cite{watkins2017inference,kohler2018dynamic,watkins2019robust,nowzari2015optimal, pare2015stability}.

In evolutionary dynamics, revision protocols dictate the mechanisms through which an individual's strategy evolves \cite{sandholm2010population}. Various dynamics, including but not limited to replicator, Smith and Brown-von Neumann-Nash dynamics, define the learning protocols that model the process of updating strategies \cite{hofbauer2003evolutionary}. Most prior works have modeled individuals as fully rational agents who act selfishly during pandemics \cite{9992844}, and used simple imitation as learning dynamics, \cite{paarporn2023sis,fu2010imitation,satapathi2022epidemic}. Although few works, such as \cite{11499399,10933512} have studied how logit learning averts undesirable outcomes in general feedback-evolving games, modeling individuals to possess bounded rationality during epidemics remains relatively underexplored, with \cite{maitra2023sis,hota2023learning,park2025epidemic,ye2021game} being departures. However, these papers focused neither on the explicit SIS disease-free and endemic transitions as rationality changed, nor compared with other learning rules.

 To bridge the existing gap, we consider \emph{logit choice} \cite{hofbauer2007evolution} as the learning rule to capture the evolution of population behavior, and examine the effect of rationality on infection spread. In this paper, we address the following question: \emph{can bounded rationality stabilize disease-free and endemic equilibria that are otherwise unstable under replicator dynamics?}
  We characterize the stationary equilibria of the coupled system and establish that rationality affects their stability properties. Our principal result identifies a regime in which the disease-free equilibrium is necessarily unstable under replicator and high rationality logit dynamics but becomes locally asymptotically stable under logit learning when rationality falls below a threshold (Theorem \ref{thm:dfe_stab}). We further show that bounded rationality can alter the asymptotic behavior of the coupled system by illustrating a parameter regime in which replicator dynamics exhibit sustained oscillations, whereas logit dynamics with sufficiently low rationality converge to a stable endemic equilibrium.
 
The rest of the paper is organized as follows. In Section \ref{sec:epi_model}, we introduce the coupled epidemiological and logit response dynamics. In this section, we also discuss the action space, and payoffs that individuals receive on choosing each action. Section \ref{secchar_slre} characterizes the stationary equilibrium under logit (bounded and complete rationality) and replicator learning rules. Section \ref{sec:exist_unique} establishes the existence and uniqueness of various stationary logit equilibria. We include the main results on local asymptotic stability of different equilibria in Section \ref{sec:stab&sigma_impact}, by analyzing the Jacobian matrix. We validate our analytical results in Section \ref{sec:num_sim} by illustrating (failure of) convergence to the steady-state equilibrium. Finally, we conclude the paper in Section \ref{sec:conc}, and include some promising directions for future research. 

%%%%%%%%%%%%%%%%%%%%%%%%%%%%%%%%%%%%%%%%%%%%%%%%%%%%%%%%%%%%%%%%%%%%%%%%%%%%%%%%
%%%%%%%%%%%%%%%%%%%%%%%%%%%%%%%%%%%%%%%%%%%%%%%%%%%%%%%%%%%%%%%%%%%%%%%%%%%%%%%%
\section{Coupled Epidemic and Behavioral Dynamics}
\label{sec:epi_model}
We introduce the epidemiological dynamics involving a coupling between infection spread and population behavior. We consider a large homogeneous population, such that at any instant $t$ each individual is either \emph{susceptible} or \emph{infected}. We denote the proportion of susceptible agents as $\St(t) \in [0, 1]$ and of infected agents as $\It(t) \in [0, 1]$. All individuals have the choice of following the health measure recommendations denoted by $\mathcal{C}$ (action \emph{cooperation}), or defying them denoted by $\mathcal{D}$ (action \emph{defection}). In our game setting, let time-varying $x(t) \in [0, 1]$ be the proportion of individuals who cooperate. Cooperation reduces an individual's effective contribution to disease transmission, whereas defection leaves transmission unchanged. Under homogeneous mixing, the population-level effective transmission rate is modeled as $\beta (1 - x)$. Given that $\St(t) = 1 - \It(t)$ is true at all instants, it suffices to examine only the evolution of the infected fraction, which we define as
\begin{equation}
    \dot{\It}(t) = \beta \cdot (1 - x(t)) \cdot \It(t) \cdot (1 - \It(t)) - \gamma \cdot \It(t),
    \label{eq:sis_dyn}
\end{equation}
where $\beta > 0$ and $\gamma > 0$ are the transmission and recovery rates, respectively. 
We also introduce the social state as $\zt(t) := (x(t), \It(t))$ which encodes the instantaneous state of the environment.

Individuals receive payoffs from choosing either cooperation or defection. We borrow the payoff structure used in standard feedback-evolving games, \cite{weitz2016oscillating,paarporn2023sis}. In this framework, individuals engage in pairwise interaction and obtain rewards defined by the payoff matrix

\begin{align}
    \begin{bmatrix}
        R_{\It}(t) & S_{\It}(t)\\
        T_{\It}(t) & P_{\It}(t)
    \end{bmatrix} \!=\! (1 - \It(t)) \begin{bmatrix}
        R_1 & S_1\\
        T_1 & P_1
    \end{bmatrix} + \It(t) \begin{bmatrix}
        R_0 & S_0\\
        T_0 & P_0
    \end{bmatrix}
    \label{eq:payoff}
\end{align}
where the first row is the payoff that a cooperator receives on interacting with a cooperator (first column) or a defector (second column), and the elements of the second row are the payoffs received by a defector for playing against a cooperator (first column) or a defector (second column). Although we do not impose any restrictions on the structure of $Q_0$, throughout the paper we assume that the payoffs satisfy $T_1 - R_1 > 0$ and $P_1 - S_1 > 0$. This is a logical assumption that arises from the reasoning that defection strictly dominates cooperation in the disease-free environment.

Henceforth, we suppress the dependence on time $t$ for better readability, and explicitly include it whenever necessary. At a social state $\zt$, the average payoffs received by a cooperator and a defector, respectively are given by
\begin{align}
    \ut_{\ct}(\zt) &:= x \cdot R_{\It} + (1 - x) \cdot S_{\It}, \nonumber
    \\ \ut_{\dt}(\zt) &:= x \cdot T_{\It} + (1 - x) \cdot P_{\It}.
    \label{eq:avg:payoff}
\end{align}
In order to incorporate bounded rationality, we model the behavioral dynamics according to continuous-time logit learning, given by
\begin{equation}
    \dot{x}(\zt; \sigma) = \frac{e^{\sigma \ut_{\ct}(\zt)}}{e^{\sigma \ut_{\ct}(\zt)} + e^{\sigma \ut_{\dt}(\zt)}} - x,
    \label{eq:logit_dyn}
\end{equation}
where $\sigma \geq 0$ is the parameter that captures the rationality of an individual. Note that $\sigma = 0$ implies that an individual is completely irrational, and as $\sigma$ increases she chooses the action that maximizes her payoff with a greater probability. 
Before we define and characterize stationary logit equilibrium, we establish the positive invariance of $\Upgamma := [0, 1]^2$ in the following lemma.
\begin{lemma}
    The set $\Upgamma := [0, 1]^2$ is positively invariant with respect to \eqref{eq:sis_dyn} and \eqref{eq:logit_dyn}.
    \label{lemm:pos_inv}
\end{lemma}
\begin{proof}
    The proof is presented in the Appendix.
\end{proof}

\section{Stationary Logit Equilibrium Characterization}
\label{secchar_slre}

We begin this section by presenting a formal definition of stationary logit equilibrium (SLE) of \eqref{eq:sis_dyn} and \eqref{eq:logit_dyn}.

\subsection{SLE Definition}
\begin{definition}
    For a given rationality parameter $\sigma \in [0, \infty]$, a social state $\zt^*=(x^*, \It^*)$ is called a stationary logit equilibrium if both
    \begin{align}
        &\dot{\It}(x^*, \It^*) = 0 \nonumber \textnormal{ and}
        \\ &\dot{x}(x^*, \It^*; \sigma) = 0 \nonumber
    \end{align}
    simultaneously hold, where for finite $\sigma < \infty$, the behavioral dynamics are given by the logit response, while $\sigma \rightarrow \infty$ denotes the complete rationality limit.
\label{def:slre}
\end{definition}
We compute payoff difference $\Delta \ut(\zt)$ from \eqref{eq:payoff} and \eqref{eq:avg:payoff} as
\begin{align}
    \Delta \ut(\zt) &:= \ut_{\dt}(\zt) - \ut_{\ct}(\zt) \nonumber
    \\ &\;\;= (1 - \It) x \delta_{TR1} + \It x \delta_{TR0} \nonumber
    \\ &\quad \;\;\; +(1 - \It) (1 - x) \delta_{PS1} + \It (1 - x) \delta_{PS0},
    \label{eq:diff_payoff}
\end{align}
where $\delta_{GFn} = G_n - F_n$, such that $G, F \in \{P, R, S, T\}$ and $n = \{0, 1\}$.
From Definition \ref{def:slre}, the solutions of
\begin{align}
    &\beta \cdot (1 - x^*) \cdot \It^* \cdot (1 - \It^*) - \gamma \cdot \It^* = 0, \textnormal{ and}
    \label{eq:I_dot}
    \\ &x^*(\sigma) = \frac{1}{1 + e^{\sigma \Delta \ut(\zt^*)}}
    \label{eq:logit_eq_pt}
\end{align}
define the SLE $\zt^* = (x^*, \It^*)$. The behavioral component of an SLE ($\zt^*$) is captured by $x^*$, whereas, $\It^*$ encodes the infection state of the population.

\subsection{Behavioral Characterization of SLE}

 We define the complete-rationality limit as the pointwise limit of the logit response \eqref{eq:logit_dyn} as $\sigma \rightarrow \infty$. At payoff indifference of $\Delta \ut = 0$, this convention selects the equally mixed response $x^* = 0.5$, whereas the set-valued best-response correspondence permits any $x^* \in [0, 1]$. The two coincide whenever $\Delta \ut \neq 0$. 
\begin{proposition}
    Suppose $\zt^*$ is an SLE of the system. When agents possess bounded rationality and additionally $\sigma > 0$ holds, the behavioral component of SLE is characterized as
    \begin{align}
       &\Delta \ut(\zt^*) > 0 \iff 0 < x^* < 0.5, \nonumber %1 - \frac{2 \gamma}{\beta} < \It^* < 1 - \frac{\gamma}{\beta}; \nonumber
       \\&\Delta \ut(\zt^*) < 0 \iff 1 > x^* > 0.5, \nonumber %0 \leq \It^*  < 1 - \frac{2 \gamma}{\beta}; \nonumber
       \\&\Delta \ut(\zt^*) = 0 \iff x^* = 0.5.%, \It^* = 1 - \frac{2 \gamma}{\beta}.
       \tag{$L^{\sigma}_{\mathtt{finite}}$} \label{eq:sigma<<infyFPP}
    \end{align}
    \label{prop:finite_sigma_x^*}
\end{proposition}
\begin{proof}
    The finite-$\sigma$ behavioral SLE characterization follows by substituting $\Delta \ut(\zt^*) > 0$, or $\Delta \ut(\zt^*) < 0$, or $\Delta \ut(\zt^*) = 0$ in \eqref{eq:logit_eq_pt}.
\end{proof}

The following proposition characterizes behavioral strategy of the equilibrium social state under complete rationality.
\begin{proposition}
    The behavioral characterization of SLE under complete rationality, i.e., in the limiting case of $\sigma \rightarrow \infty$, is different from Proposition \ref{prop:finite_sigma_x^*}. The characterization is
    \begin{align}
       &\Delta \ut(\zt^*) > 0 \iff x^* = 0, \nonumber  \nonumber
       \\&\Delta \ut(\zt^*) = 0 \iff x^* = 0.5.
       \tag{$L^{\sigma}_{\infty}$} \label{eq:sigma=infyFPP}
    \end{align}
    \label{rem:x^*_char_sigma=infty}
\end{proposition}
We exclude $\Delta \ut(\zt^*) < 0$ because when payoff difference is strictly negative, under complete rationality $\zt^* = (x^* = 1, \It^* = 0)$ is the only admissible SLE, whereas $\Delta \ut(x^* = 1, \It^* = 0) = \delta_{TR1}$ is strictly positive, which contradicts the premise $\Delta \ut(\zt^*) < 0$. This implies that $(1, 0)$ is not an SLE. 

We now categorize SLE based on characteristics of $x^*$ and $\It^*$, where $\It^*$ can either be $0$, or equal to $1 - \frac{\gamma}{\beta (1 - x^*)}$.
\begin{itemize}
    \item Disease-free SLE $\zt^*_{\mathtt{DFE}} := (x^*_{\mathtt{DFE}}, \It^*_{\mathtt{DFE}})$, where $\It^*_{\mathtt{DFE}} = 0$,
    \item Interior endemic SLE $\zt^*_{\mathtt{EE(int)}} := (x^*_{\mathtt{EE(int)}}, \It^*_{\mathtt{EE(int)}})$, where $x^*_{\mathtt{EE(int)}} = 0.5, \It^*_{\mathtt{EE(int)}} = 1 - \frac{2 \gamma}{\beta} > 0$,
    \item Finite-$\sigma$ endemic SLE $\zt^*_{\mathtt{EE(\sigma)}} := (x^*_{\mathtt{EE(\sigma)}}, \It^*_{\mathtt{EE(\sigma)}})$, where $x^*_{\mathtt{EE(\sigma)}} \in (0, 1) \setminus \{0.5 \}, \It^*_{\mathtt{EE(\sigma)}} = 1 - \frac{\gamma}{\beta (1 - x^*_{\mathtt{EE(\sigma)}})} > 0$,
    \item Corner endemic SLE $\zt^*_{\mathtt{EE(cor)}} := (x^*_{\mathtt{EE(cor)}}, \It^*_{\mathtt{EE(cor)}})$, where $x^*_{\mathtt{EE(cor)}} = 0, \It^*_{\mathtt{EE(cor)}} = 1 - \frac{\gamma}{\beta} > 0$.
\end{itemize}
Note that corner endemic SLE $\zt^*_{\mathtt{EE(cor)}}$ does not exist under bounded rationality and finite-$\sigma$ endemic SLE $\zt^*_{\mathtt{EE(\sigma)}}$ does not exist when individuals are completely rational. 

In order to compare with the popular imitation learning, we characterize $x^*$ when individuals update their strategies in accordance with \emph{replicator dynamics} \cite{sandholm2010population}.
\begin{remark}
    The characterizations in \eqref{eq:sigma<<infyFPP} and \eqref{eq:sigma=infyFPP} differ from that of replicator dynamics defined as
    \begin{equation}
        \dot{x} = x (1 - x) (-\Delta \ut(\zt)),
        \label{eq:repl_dyn}
    \end{equation}
    for which the stable behavioral states satisfy
    \begin{align}
    \Delta\mathcal{U}(\mathbf{z}^*) > 0 &\implies x^* = 0, \nonumber
    \\ \Delta\mathcal{U}(\mathbf{z}^*) = 0 &\implies x^* \in [0, 1). \tag{R}
    \label{eq:replicatorFPP}
\end{align}
Here, we exclude $x^* = 1$ since, although $\dot{x} = 0$ holds at $x = 1$, the state $(1,0)$ is unstable as $\Delta\mathcal{U}(1,0) = \delta_{TR1} > 0$.
The exact fixed-point social state $(x^*, \It^*)$ in the second case of \eqref{eq:replicatorFPP} is obtained by solving for $\Delta \ut\big(x^*, 1 - \frac{\gamma}{\beta (1 - x^*)}\big) = 0$. In other words, the interior stationary point of replicator dynamics is not necessarily always $x^* = 0.5$, which is in contrast to logit dynamics. 
\label{rem:repli_dyn}
\end{remark}

We analyze the closed-loop system \eqref{eq:sis_dyn} and \eqref{eq:logit_dyn} and study the existence of each SLE in the following section.

\section{Existence of SLE}
\label{sec:exist_unique}

We begin this section by stating the following proposition on existence and local asymptotic stability of an infection equilibrium defined by system \eqref{eq:sis_dyn}.
\begin{proposition}
    Let fixed $x^*_{\mathtt{DFE}}, x^*_{\mathtt{EE}} \in [0, 1]$ be a solution to equation \eqref{eq:logit_eq_pt}. Then the epidemic dynamics \eqref{eq:sis_dyn} has two equilibria:
    \begin{itemize}
        \item the disease-free equilibrium $\It^*_{\mathtt{DFE}} = 0$, which always exists, and is locally asymptotically stable (LAS) only when $\gamma \geq \beta (1 - x^*_{\mathtt{DFE}})$, and
        \item the endemic equilibrium $\It^*_{\mathtt{EE}} = 1 - \frac{\gamma}{\beta (1 - x^*_{\mathtt{EE}})}$, which is LAS and exists only when $\gamma < \beta (1 - x^*_{\mathtt{EE}})$.
    \end{itemize}
    \label{prop:exist_sis}
\end{proposition}
The above result is analogous to the result obtained in the discrete-time version of the system examined in \cite{maitra2023sis}, and hence the proof is omitted in the interest of space. Characterization of the (completely rational) population strategy when infection is completely absent in steady-state is presented in the following remark.  
\begin{remark}
    Proposition \ref{prop:exist_sis} states that the disease-free branch is LAS when $1\geq x^*_{\mathtt{DFE}}\geq 1-\frac{\gamma}{\beta}$. Since $\delta_{TR1}, \delta_{PS1} > 0$, defection is the unique strictly dominant action in the disease-free regime under complete rationality. Hence, the equilibrium strategy is $x^*_{\mathtt{DFE}}=0$, and the corresponding SLE is $(0, 0)$. However, we later show that when the basic reproduction number $\mathcal{R}_0 := \frac{\beta}{\gamma} > 1$, SLE $(0, 0)$ is unstable.
    \label{rem:X^*=0_dominant_I^*=0}
\end{remark}
 
 We now derive conditions under which each SLE exists.

\subsection{Existence of SLE $\zt^*_{\mathtt{DFE}}$}

Proposition \ref{prop:exist_sis} establishes that a disease-free equilibrium $\zt^*_{\mathtt{DFE}}$ always exists. We now tackle the next SLE.

\subsection{Existence of SLE $\zt^*_{\mathtt{EE(int)}}$}

We state the lemma on the existence and uniqueness of interior endemic equilibrium.

\begin{lemma}
    Suppose the basic reproduction number $\mathcal{R}_{0} := \frac{\beta}{\gamma} > 2$ holds. Then for all $\sigma \geq 0$, the interior endemic SLE $\zt^*_{\mathtt{EE(int)}} := (x^*_{\mathtt{EE(int)}} = 0.5, \It^*_{\mathtt{EE(int)}} = 1 - \frac{2 \gamma}{\beta})$ exists when
    \begin{equation}
        2 \gamma (\delta_{TR1} + \delta_{PS1}) + (\beta - 2 \gamma) (\delta_{TR0} + \delta_{PS0}) = 0
        \label{eq:nc_x^*=0.5}
    \end{equation}
    is satisfied. Furthermore, suppose the following holds:
    \begin{align}
    \min &\bigg\{\delta_{TR0} - \delta_{PS0} + \frac{2 \gamma}{\beta} (\delta_{TR1} - \delta_{TR0}), \nonumber
    \\ &\qquad \qquad 2 \delta_{TR1} - \delta_{TR0} - \delta_{PS0} \bigg\} \geq 0.
        \label{eq:cond1}
    \end{align}
    Then $\zt^*_{\mathtt{EE(int)}}$ is the unique endemic SLE.
    \label{lemm:x^*=0.5_exist_unique}
\end{lemma}
\begin{proof}
    Note that the constraint on reproduction number $\mathcal{R}_0 > 2$ is necessary for the existence of $\It^*_{\mathtt{EE(int)}}$. This is because, $\mathcal{R}_0 < 2 \Rightarrow \frac{\beta}{\gamma} < 2$, which in turn implies that $\frac{2 \gamma}{\beta} > 1$. Since the steady-state infected fraction at $x^*_{\mathtt{EE(int)}} = 0.5$ is $1 > \It^*_{\mathtt{EE(int)}} = 1 - \frac{2 \gamma}{\beta} > 0$, the endemic equilibrium ceases to exist unless $\mathcal{R}_0 > 2$ holds. First, we derive a condition under which SLE $(0.5, 1 - \frac{2 \gamma}{\beta})$ exists, i.e., $\Delta \ut(\zt^*_{\mathtt{EE(int)}}) = 0$ holds. Substituting $x^*_{\mathtt{EE(int)}} \!\!=\! 0.5, \It^*_{\mathtt{EE(int)}} \!\! = \! 1 \! - \! \frac{2 \gamma}{\beta}$ in \eqref{eq:diff_payoff} yields $$ 2 \gamma (\delta_{TR1} + \delta_{PS1}) + (\beta - 2 \gamma) (\delta_{TR0} + \delta_{PS0}) = 0.$$

    We now prove the uniqueness of the SLE. To this end, we define a function $H(x^*_{\mathtt{EE}})$ on the endemic branch as
    \begin{equation}
    H(x^*_{\mathtt{EE}}) = \delta_{TR0} - \delta_{PS0} + \frac{2 \gamma}{\beta (1 - x^*_{\mathtt{EE}})} (\delta_{TR1} - \delta_{TR0}).
    \label{eq:H_x}
\end{equation}
    At any endemic equilibrium $(x^*_{\mathtt{EE}}, \It^*_{\mathtt{EE}})$, the following hold.
    \begin{align}
        0.5(2 x^*_{\mathtt{EE}} &- 1) H(x^*_{\mathtt{EE}}) = (x^*_{\mathtt{EE}} - 0.5)(\delta_{TR0}- \delta_{PS0}) \nonumber
        \\ &\quad+ \frac{\gamma (2 x^*_{\mathtt{EE}} - 1)}{\beta (1 - x^*_{\mathtt{EE}})} (\delta_{TR1}- \delta_{TR0}), \label{eq:0.5(2x-1)H}
        \\ \Delta\ut \big(x^*_{\mathtt{EE}}, &1 - \frac{\gamma}{\beta (1 - x^*_{\mathtt{EE}})}\big) = x^*_{\mathtt{EE}} \delta_{TR0} + (1 - x^*_{\mathtt{EE}}) \delta_{PS0} \nonumber
        \\ &\quad + \frac{\gamma}{\beta} (\delta_{PS1} \! - \! \delta_{PS0}) \! + \! \frac{\gamma x^*_{\mathtt{EE}}}{\beta (1 - x^*_{\mathtt{EE}})} (\delta_{TR1} \! - \! \delta_{TR0}).
        \label{eq:delta_U=H}
    \end{align}
    Now, subtracting \eqref{eq:0.5(2x-1)H} from \eqref{eq:delta_U=H} yields
    \begin{align*}
        \Delta\ut &\big(x^*_{\mathtt{EE}}, 1 - \frac{\gamma}{\beta (1 - x^*_{\mathtt{EE}})}\big) - 0.5(2 x^*_{\mathtt{EE}} - 1) H(x^*_{\mathtt{EE}})
        \\ & = 0.5 (\delta_{TR0} \! + \! \delta_{PS0}) \! + \! \frac{\gamma}{\beta} (\delta_{TR1} \!+ \! \delta_{PS1} \! - \delta_{TR0} \! - \! \delta_{PS0}),
        %\\ &= 0,
    \end{align*}
    which is equal to zero from equation \eqref{eq:nc_x^*=0.5}. Thus, when \eqref{eq:nc_x^*=0.5} is satisfied, we have
    \begin{equation}
        \Delta \ut(x^*_{\mathtt{EE}}, 1 - \frac{\gamma}{\beta (1 - x^*_{\mathtt{EE}})}) = 0.5 (2 x^*_{\mathtt{EE}} - 1) H(x^*_{\mathtt{EE}}).
        \label{eq:del_u_H}
    \end{equation}
    The sign of $H'(x^*_{\mathtt{EE}}) = \frac{2 \gamma}{\beta (1 - x^*_{\mathtt{EE}})^2} (\delta_{TR1} - \delta_{TR0})$ remains constant on the endemic branch, since it is monotonic in $x^*_{\mathtt{EE}}$. If $\delta_{TR1} \geq \delta_{TR0}$ holds, then $H(x^*_{\mathtt{EE}})$ is non-decreasing in $x^*_{\mathtt{EE}}$, and minimum occurs at $x^*_{\mathtt{EE}} = 0$, i.e., $H(0) = \delta_{TR0} - \delta_{PS0} + \frac{2 \gamma}{\beta} (\delta_{TR1} - \delta_{TR0})$ which is non-negative by \eqref{eq:cond1}.

    If instead $\delta_{TR1} < \delta_{TR0}$ holds, then $H(x^*_{\mathtt{EE}})$ is decreasing, and the infimum occurs at $x^*_{\mathtt{EE}} \rightarrow 1 - \frac{\gamma}{\beta}$, i.e., $H(1 - \frac{\gamma}{\beta}) = 2 \delta_{TR1} - \delta_{TR0} - \delta_{PS0} \geq 0$ which is again non-negative by \eqref{eq:cond1}. Thus, in either case $H(x^*_{\mathtt{EE}}) \geq 0$ along the endemic branch. From \eqref{eq:del_u_H}, we obtain the relations,
    \begin{align*}
        x^*_{\mathtt{EE}} < 0.5 &\implies \Delta \ut(x^*_{\mathtt{EE}}, 1 - \frac{\gamma}{\beta (1 - x^*_{\mathtt{EE}})}) \leq 0,
        \\x^*_{\mathtt{EE}} > 0.5 &\implies \Delta \ut(x^*_{\mathtt{EE}}, 1 - \frac{\gamma}{\beta (1 - x^*_{\mathtt{EE}})}) \geq 0
    \end{align*}
    which contradicts the characterization \eqref{eq:sigma<<infyFPP} in Proposition \ref{prop:finite_sigma_x^*}. Therefore we conclude that $z^*_{\mathtt{EE(int)}}$ is the only endemic equilibrium that exists.
\end{proof}
The existence condition derived in \eqref{eq:nc_x^*=0.5} also determines the sign of the sum $\delta_{TR0} + \delta_{PS0}$. Since $\beta - 2 \gamma > 0$ and $\delta_{TR1} + \delta_{PS1} > 0$ are true, it is necessary that $\delta_{TR0} + \delta_{PS0} < 0$ holds for the existence of SLE $\zt^*_{\mathtt{EE(int)}}$.

\subsection{Existence of SLE $\zt^*_{\mathtt{EE(\sigma)}}$}
We now prove the existence and uniqueness pertaining to the finite-$\sigma$ endemic equilibrium.
\begin{lemma}
    Let the basic reproduction number satisfy $1 < \mathcal{R}_0 \leq 2$, or $\mathcal{R}_0 > 2$ with violation of equality \eqref{eq:nc_x^*=0.5}. Then the finite-$\sigma$ endemic SLE $\zt^*_{\mathtt{EE(\sigma)}} := (x^*_{\mathtt{EE(\sigma)}}, \It^*_{\mathtt{EE(\sigma)}})$ exists when $\infty > \sigma > \max \big\{0, \sigma^{\mathtt{TH}} \big\}$ holds, where
    \begin{align}
        &\sigma^{\mathtt{TH}} := \frac{1}{(1 - \frac{\gamma}{\beta}) \delta_{TR1} + \frac{\gamma}{\beta} \delta_{PS1}} \times \mathtt{ln} (\frac{\gamma}{\beta - \gamma}). \label{eq:sigma^th_1}
    \end{align}
    Furthermore, suppose
    \begin{align}
        \min \bigg\{ \gamma (\delta_{TR0} &- \delta_{PS0}) + \beta (\delta_{TR1} - \delta_{TR0}), \nonumber\\ & \beta (\delta_{TR0} - \delta_{PS0}) + \gamma (\delta_{TR1} - \delta_{TR0}) \bigg \} \geq 0,
        \label{eq:min_suff_cond_x^*=x^*_sigmaEE}
    \end{align}
    also holds. Then $\zt^*_{\mathtt{EE(\sigma)}}$ is the only endemic SLE that exists.
    \label{lemm:x^*=x^*_ee}
\end{lemma}
\begin{proof}
    Recall that interior SLE $\zt^*_{\mathtt{EE(int)}}$ does not exist when \eqref{eq:nc_x^*=0.5} is not satisfied. We define  
    \begin{equation}
        \mathcal{F}(x) = \frac{1}{1 + e^{\sigma \Delta \ut\bigg(x, 1 - \frac{\gamma}{\beta (1 - x)} \bigg)}} - x,
        \label{eq:F(x*)}
    \end{equation}
    where $x$ lies in the interval $[\eta, 1 - \frac{\gamma}{\beta}]$ for a small $\eta > 0$, such that $\sigma \rightarrow \infty \Rightarrow \eta \rightarrow 0$. Function $\mathcal{F}(x)$ is continuous on $[\eta, 1 - \frac{\gamma}{\beta}]$. It can be shown that for every $\sigma$ there exists a small $\eta$ such that $\mathcal{F}(\eta) > 0$ holds. At $x = 1 - \frac{\gamma}{\beta}$, steady-state infected proportion is zero and the payoff difference is $$\Delta \ut(x, 0) = (1 - \frac{\gamma}{\beta}) \delta_{TR1} + \frac{\gamma}{\beta} \delta_{PS1} > 0.$$ Now, under \eqref{eq:min_suff_cond_x^*=x^*_sigmaEE} it is necessary and sufficient that $\mathcal{F}(1 - \frac{\gamma}{\beta})$ is strictly negative for the existence of a finite-$\sigma$ endemic equilibrium. Solving for $\mathcal{F}(1 - \frac{\gamma}{\beta}) < 0$ results in
    \begin{align}
        &\frac{1}{1 + e^{\sigma \Delta \ut\big(1 - \frac{\gamma}{\beta}, 0\big)}} - \bigg(1 - \frac{\gamma}{\beta}\bigg) < 0 \nonumber
        \\ \Rightarrow &e^{\sigma \Delta \ut\big(1 - \frac{\gamma}{\beta}, 0\big)} > \frac{\gamma}{\beta - \gamma} \nonumber
        \\ \Rightarrow &\sigma > \frac{1}{(1 - \frac{\gamma}{\beta}) \delta_{TR1} + \frac{\gamma}{\beta} \delta_{PS1}} \times \mathtt{ln} (\frac{\gamma}{\beta - \gamma}) =: \sigma^{\mathtt{TH}}, \nonumber
    \end{align}
    where we obtained the final inequality by substituting $\Delta \ut(1 - \frac{\gamma}{\beta}, 0) = (1 - \frac{\gamma}{\beta}) \delta_{TR1} + \frac{\gamma}{\beta} \delta_{PS1}$. Thus, $\sigma > \sigma^{\mathtt{TH}}$ guarantees the existence of at least one endemic SLE.

    To prove the uniqueness of $\zt^*_{\mathtt{EE(\sigma)}}$, we proceed as follows. Differentiating \eqref{eq:delta_U=H} we obtain
    \begin{equation}
        \frac{d \Delta \ut(\zt)}{d x} = \delta_{TR0} - \delta_{PS0} + \frac{\gamma (\delta_{TR1} - \delta_{TR0})}{\beta (1 - x)^2}.
        \label{eq:diff_U_x^*}
    \end{equation}
    The R.H.S of \eqref{eq:diff_U_x^*} is monotonic in $x$ and the sign does not change in the range of $x \in [\eta, 1 - \frac{\gamma}{\beta}]$. Furthermore, condition \eqref{eq:min_suff_cond_x^*=x^*_sigmaEE} ensures that both 
    $$\frac{d \Delta \ut(0)}{d x} = \frac{1}{\beta} \big(\beta (\delta_{TR0} - \delta_{PS0}) + \gamma (\delta_{TR1} - \delta_{TR0})\big) \geq 0, \textnormal{ and}$$ 
    $$\frac{d \Delta \ut(1 - \frac{\gamma}{\beta})}{d x} = \frac{1}{\gamma} \big(\gamma (\delta_{TR0} - \delta_{PS0}) + \beta (\delta_{TR1} - \delta_{TR0})\big) \geq 0,$$
    are satisfied, implying that $\frac{d \Delta \ut(x)}{d x} \geq 0$ in the entire interval of $[\eta, 1 - \frac{\gamma}{\beta}]$. Differentiating $\mathcal{F}(x)$ in \eqref{eq:F(x*)} w.r.t $x$ yields
    $$\mathcal{F}'(x) = -\frac{\sigma e^{\sigma \Delta \ut(x)}}{(1 + e^{\sigma \Delta \ut(x)})^2} \times \underbrace{\frac{d \Delta \ut(x)}{d x}}_{\geq 0} - 1,$$ which is strictly negative for all $x \in [\eta, 1 - \frac{\gamma}{\beta}]$. Since $\mathcal{F}(x)$ is strictly decreasing in $x$, it implies that existing finite-$\sigma$ SLE $\zt^*_{\mathtt{EE(\sigma)}}$ is unique. This concludes our proof.
\end{proof}

\subsection{Existence of SLE $\zt^*_{\mathtt{EE(cor)}}$}

We finally state the existence of the corner endemic SLE in the limiting case of $\sigma \rightarrow \infty$ as the following lemma.
\begin{lemma}
    Suppose the basic reproduction number satisfies $\mathcal{R}_0 > 1$. When individuals are completely rational, the corner endemic SLE $\zt^*_{\mathtt{EE(cor)}} := (x^*_{\mathtt{EE(cor)}} = 0, \It^*_{\mathtt{EE(cor)}} = 1 - \frac{\gamma}{\beta})$ exists when the inequality
    \begin{equation}
        \delta_{PS1} + \bigg(\frac{\beta}{\gamma} - 1 \bigg) \delta_{PS0} > 0
        \label{eq:delta2}
    \end{equation}
    holds. Furthermore, suppose
    \begin{equation}
        \min \bigg\{ \delta_{TR1} - \delta_{PS0}, \beta (\delta_{TR0} - \delta_{PS0}) + \gamma (\delta_{TR1} - \delta_{TR0}) \bigg \} > 0,
        \label{eq:min_suff_cond_x^*=0}
    \end{equation}
    also holds. Then $\zt^*_{\mathtt{EE(cor)}}$ is the only endemic SLE that exists.
    \label{lemm:x^*=0}
\end{lemma}
\begin{proof}
   The proof is presented in the Appendix.
\end{proof}
Although our modeling assumption ensures that $\delta_{TR1}, \delta_{PS1} > 0$ hold, inequalities \eqref{eq:cond1}, \eqref{eq:min_suff_cond_x^*=x^*_sigmaEE} and \eqref{eq:min_suff_cond_x^*=0} introduce additional incentive conditions involving the game and disease parameters. While these sufficient conditions explore a part of the parameter space to guarantee the uniqueness of different SLE, characterizing remaining spaces is beyond the scope of this paper and remains an interesting direction for future work.

\section{Stability of SLE and Impact of Rationality}
\label{sec:stab&sigma_impact}

In this section, we present some of our main results. First, we show that there exist parameter regimes in which the disease-free equilibrium is unstable under imitative learning and under high-rationality logit learning, whereas a sufficiently low level of rationality stabilizes it. We also prove that in certain parameter regimes an unstable endemic equilibrium stabilizes when rationality falls below a threshold and analytically derive the threshold.

We begin by analyzing the local asymptotic stability of the disease-free equilibrium. To this end, we define the \emph{Jacobian} $\mathcal{J}$ at equilibrium $\zt^*$ as
\begin{align}
    &\mathcal{J}(\zt^*) := \begin{bmatrix}
        \frac{\partial \dot{\It}(\zt^*)}{\partial \It} & \frac{\partial \dot{\It}(\zt^*)}{\partial x}\\
        \frac{\partial \dot{x}(\zt^*)}{\partial \It} & \frac{\partial \dot{x}(\zt^*)}{\partial x}
    \end{bmatrix}, \nonumber
    %\\ &\textnormal{where } \nonumber 
    \\ &\frac{\partial \dot{\It}(\zt^*)}{\partial \It} = \beta (1 - x^*) (1 - 2 \It^*) - \gamma, 
    \label{eq:idot_i}
    \\ &\frac{\partial \dot{\It}(\zt^*)}{\partial x} = -\beta \It^* (1 - \It^*),
    \label{eq:idot_x}
    \\ &\frac{\partial \dot{x}(\zt^*)}{\partial \It} = -\frac{\sigma e^{\sigma \Delta \ut(\zt^*)}}{(1 + e^{\sigma \Delta \ut(\zt^*)})^2} \times \frac{\partial \Delta \ut(\zt^*)}{\partial \It},
    \label{eq:xdot_i}
    \\ &\frac{\partial \dot{x}(\zt^*)}{\partial x} = -\frac{\sigma e^{\sigma \Delta \ut(\zt^*)}}{(1 + e^{\sigma \Delta \ut(\zt^*)})^2} \times \frac{\partial \Delta \ut(\zt^*)}{\partial x} - 1.
    \label{eq:xdot_x}
\end{align}
Note that the disease-free equilibrium $\zt^*_{\mathtt{DFE}}$ is stable for all $\sigma \geq 0$ when $\gamma > \beta$ and $\delta_{TR1} \geq \delta_{PS1}$ hold. However, we consider a special case of $\gamma < \beta < 2 \gamma$, in which we show that both, high degree of rationality under logit dynamics and imitation learning under replicator dynamics result in the instability of $\zt^*_{\mathtt{DFE}}$, whereas a low rationality stabilizes the equilibrium.
\subsection{LAS of $\zt^*_{\mathtt{DFE}}$}
\begin{theorem}
    Suppose $1 < \mathcal{R}_0 < 2$ holds. In addition, let $\delta_{TR1} \geq \delta_{PS1}$ also hold. The disease-free SLE $\zt^*_{\mathtt{DFE}} = (x^*_{\mathtt{DFE}}(\sigma), 0)$ is LAS for $0 \leq \sigma < \sigma^{\mathtt{TH}}$ and unstable for $\sigma > \sigma^{\mathtt{TH}}$, where $\sigma^{\mathtt{TH}}$ is defined in \eqref{eq:sigma^th_1}, and
    \begin{align}
         &x^*_{\mathtt{DFE}}(\sigma) := \frac{1}{1 + e^{\sigma \Delta \ut(x^*_{\mathtt{DFE}}(\sigma), 0)}}. \label{eq:x^*DFE_sigma}
    \end{align}    
    \label{thm:dfe_stab}
\end{theorem}
\begin{proof}
    Recall from the previous section that SLE $\zt^*_{\mathtt{DFE}}$ always exists, and from Proposition \ref{prop:exist_sis} that $\zt^*_{\mathtt{DFE}}$ is stable if and only if $$\gamma \geq \beta (1 - x^*_{\mathtt{DFE}}) \implies x^*_{\mathtt{DFE}} \geq 1 - \frac{\gamma}{\beta}$$ holds. Substituting $x^*_{\mathtt{DFE}} = 1 - \frac{\gamma}{\beta}$, $\Delta \ut(\zt^*_{\mathtt{DFE}}) = (1 - \frac{\gamma}{\beta}) \delta_{TR1} + \frac{\gamma}{\beta} \delta_{PS1}$ and $\sigma = \sigma^{\mathtt{TH}}$ in \eqref{eq:logit_eq_pt}, we obtain
    \begin{align}
        &\frac{1}{1 + e^{\sigma^{\mathtt{TH}} \big((1 - \frac{\gamma}{\beta}) \delta_{TR1} + \frac{\gamma}{\beta} \delta_{PS1}\big)}} = 1 - \frac{\gamma}{\beta} \nonumber
        \\ \implies &\sigma^{\mathtt{TH}} = \frac{1}{(1 - \frac{\gamma}{\beta}) \delta_{TR1} + \frac{\gamma}{\beta} \delta_{PS1}} \times \mathtt{ln} (\frac{\gamma}{\beta - \gamma}). \nonumber
    \end{align}
    Note that since $\sigma^{\mathtt{TH}}$ is positive, the above expression is valid when $\frac{\gamma}{\beta - \gamma} > 1 \iff \mathcal{R}_0 < 2$ holds.

    The expression $\frac{\partial \dot{\It}(\zt^*_{\mathtt{DFE}})}{\partial x}$ in \eqref{eq:idot_x} is zero at $\It^*_{\mathtt{DFE}} = 0$. Thus, $\mathcal{J}(\zt^*_{\mathtt{DFE}})$ at the disease-free equilibrium reduces to a lower triangular matrix, implying that its diagonal elements 
    $$\frac{\partial \dot{\It}(\zt^*_{\mathtt{DFE}})}{\partial \It} \textnormal{ and } \frac{\partial \dot{x}(\zt^*_{\mathtt{DFE}})}{\partial x}$$
    are the eigenvalues. Computing $\frac{\partial \Delta \ut(\zt^*)}{\partial x}$ from \eqref{eq:diff_payoff}, we obtain $$\frac{\partial \Delta \ut(x^* = x^*_{\mathtt{DFE}}(\sigma), \It^* = 0)}{\partial x} = \delta_{TR1} - \delta_{PS1}.$$ It is easy to see that at an equilibrium strategy of $x^* = x^*_{\mathtt{DFE}}(\sigma)$ the eigenvalue 
    \begin{equation}
        \frac{\partial \dot{x}(\zt^*_{\mathtt{DFE}})}{\partial x} = \! - \! \frac{\sigma e^{\sigma \Delta \ut(\zt^*_{\mathtt{DFE}})}}{(1 + e^{\sigma \Delta \ut(\zt^*_{\mathtt{DFE}})})^2} \times (\delta_{TR1} \! - \! \delta_{PS1}) \! - \! 1
        \label{eq:del_xdot_del_x^*_dfe}
    \end{equation}
    in \eqref{eq:xdot_x} is strictly negative for $\sigma \geq 0$ under the sufficient condition of $\delta_{TR1} - \delta_{PS1}$ being non-negative. The other eigenvalue
    \begin{equation}
        \frac{\partial \dot{\It}(x_{\mathtt{DFE}}(\sigma), 0)}{\partial \It} = \beta (1 - x^*_{\mathtt{DFE}}(\sigma)) - \gamma
        \label{eq:del_idot_del_i^*_dfe}
    \end{equation} 
    is strictly negative for $x^*_{\mathtt{DFE}}(\sigma) > 1 - \frac{\gamma}{\beta}$, equal to zero for $x^*_{\mathtt{DFE}}(\sigma) = 1 - \frac{\gamma}{\beta}$, and is positive when $x^*_{\mathtt{DFE}}(\sigma) < 1 - \frac{\gamma}{\beta}$ is true. It remains to show that $x^*_{\mathtt{DFE}}(\sigma)$ decreases monotonically with $\sigma$, and as the proportion $x^*_{\mathtt{DFE}}(\sigma)$ falls below $1 - \frac{\gamma}{\beta}$, the eigenvalue $\frac{\partial \dot{\It}(x^*_{\mathtt{DFE}}(\sigma), \It^* = 0)}{\partial \It}\bigg|_{\sigma > \sigma^{\mathtt{TH}}}$ becomes positive.
    
    In order to prove that $\frac{\partial \dot{\It}(x^*_{\mathtt{DFE}}(\sigma), \It^* = 0)}{\partial \It} = \beta (1 - x^*_{\mathtt{DFE}}(\sigma)) - \gamma$ is strictly positive when $\sigma$ is higher than the threshold $\sigma > \sigma^{\mathtt{TH}}$, it is sufficient to show that equilibrium $x^*_{\mathtt{DFE}}(\sigma)$ is strictly decreasing in $\sigma$. At equilibrium since $\dot{x} = 0$ holds, equality $\frac{d \dot{x} (x^*_{\mathtt{DFE}}(\sigma), \It^*(\sigma), \sigma)}{d \sigma} = 0$ must be satisfied, which implies
    \begin{align*}
        \frac{d \dot{x} (x^*_{\mathtt{DFE}}(\sigma), \It^*(\sigma), \sigma)}{d \sigma} &= \frac{\partial \dot{x}}{\partial x^*_{\mathtt{DFE}}} \cdot \frac{d x^*_{\mathtt{DFE}}}{d \sigma} 
        \\ &\quad + \frac{\partial \dot{x}}{\partial \It^*_{\mathtt{DFE}}} \cdot \frac{d \It^*_{\mathtt{DFE}}}{d \sigma} + \frac{\partial \dot{x}}{\partial \sigma} = 0.
    \end{align*}
    Since $\It^*_{\mathtt{DFE}} = 0$ holds in the entire disease-free branch, the above equation reduces to
    \begin{align}
         &\frac{\partial \dot{x}}{\partial x^*_{\mathtt{DFE}}} \cdot \frac{d x^*_{\mathtt{DFE}}}{d \sigma} + \frac{\partial \dot{x}}{\partial \sigma} = 0 \Rightarrow \frac{d x^*_{\mathtt{DFE}}}{d \sigma} = -\frac{\partial \dot{x} / \partial \sigma}{\partial \dot{x} / \partial x^*_{\mathtt{DFE}}}.
         \label{eq:d_xdot_d_sigma}
    \end{align}
    Now, from \eqref{eq:logit_dyn} we compute the partial derivative as $$\frac{\partial \dot{x}}{\partial \sigma} = - \frac{\Delta \ut \cdot e^{\sigma \Delta \ut}}{(1 + e^{\sigma \Delta \ut})^2},$$ which is strictly negative at $\It^*_{\mathtt{DFE}} = 0$ (since, $\Delta \ut(x^*_{\mathtt{DFE}}(\sigma), 0) = x^*_{\mathtt{DFE}}(\sigma) \delta_{TR1} + (1 - x^*_{\mathtt{DFE}}(\sigma)) \delta_{PS1} > 0$). The term $\frac{\partial \dot{x}(\zt^*_{\mathtt{DFE}})}{\partial x}$ from \eqref{eq:del_xdot_del_x^*_dfe} is also negative. Therefore, $\frac{d x^*_{\mathtt{DFE}}(\sigma)}{d \sigma}$ in \eqref{eq:d_xdot_d_sigma} is strictly negative. Since $x^*_{\mathtt{DFE}}(\sigma) = 1 - \frac{\gamma}{\beta}$ holds at $\sigma = \sigma^{\mathtt{TH}}$, thus for any lower value of $\sigma$ satisfying $0 \leq \sigma < \sigma^{\mathtt{TH}}$ the equilibrium fraction satisfies $x^*_{\mathtt{DFE}}(\sigma) > 1 - \frac{\gamma}{\beta}$. Consequently, the eigenvalue $\frac{\partial \dot{\It}(x^*_{\mathtt{DFE}}(\sigma), \It^* = 0)}{\partial \It}$ in \eqref{eq:del_idot_del_i^*_dfe} becomes negative and the disease-free SLE stabilizes. When $\sigma$ exceeds the threshold $\sigma^{\mathtt{TH}}$, the fraction $x^*_{\mathtt{DFE}}(\sigma)$ reduces and $\frac{\partial \dot{\It}(x^*_{\mathtt{DFE}}(\sigma), \It^* = 0)}{\partial \It}$ becomes strictly positive, i.e., the disease-free SLE $\zt^*_{\mathtt{DFE}}$ becomes unstable. This concludes our proof.
\end{proof}

The above theorem states that bounded rationality can stabilize a disease-free equilibrium in a parameter regime in which sufficiently high rationality destabilizes it. The condition $\delta_{TR1} \geq \delta_{PS1}$ means that in a disease-free environment, the payoff advantage of defection is not lower when interacting with a cooperator than with a defector. This is justified since in the absence of an infection, the incentive of defection increases when interacting with a cooperator compared to a defector, reflecting a free-riding effect.

While we discuss the LAS of $\zt^*_{\mathtt{DFE}}$ at $\sigma \! < \! \sigma^{\mathtt{TH}}$ and instability at $\sigma \! > \! \sigma^{\mathtt{TH}}$, disease-free equilibrium is non-hyperbolic at $\sigma \! = \! \sigma^{\mathtt{TH}}$ since dominant eigenvalue is zero. Linearization does not establish LAS at a threshold that requires separate analysis, which is beyond the scope of this paper.

\begin{remark}
    We also briefly examine the disease-free equilibrium when the revision protocol is replicator dynamics. At the disease-free equilibrium, $\Delta \ut (x^*_{\mathtt{DFE}}, 0) = x^*_{\mathtt{DFE}} \delta_{TR1} + (1 - x^*_{\mathtt{DFE}}) \delta_{PS1} > 0$. The replicator dynamics in \eqref{eq:repl_dyn} drives $x^*$ to $0$. Consequently the eigenvalue $\frac{\partial \dot{\It}(0, 0)}{\partial \It} = \beta - \gamma$ is strictly positive when $1 < \mathcal{R}_0 < 2$ holds, implying that the disease-free equilibrium is unstable. While in this entire regime the disease-free equilibrium is unstable under replicator dynamics, logit dynamics at a low rationality stabilizes the equilibrium. 
\end{remark}

\subsection{LAS of $\zt^*_{{\mathtt{EE(int)}}}$}

Next, we study the LAS of the SLE $\zt^*_{{\mathtt{EE(int)}}} = (0.5, 1 - \frac{2 \gamma}{\beta})$. To investigate the stability, we introduce the following quantities.
\begin{align}
    \mathcal{A} &:= 2 \gamma (\delta_{TR1} - \delta_{PS1}) + (\beta - 2 \gamma) (\delta_{TR0} - \delta_{PS0}), \label{eq:A}
    %\\ \mathcal{B} &:= -0.5 \mathcal{A} - \frac{\beta \gamma (\delta_{TR1} + \delta_{PS1})}{\beta - 2 \gamma}, \label{eq:B}
    \\ \sigma^{\mathtt{c}} &:= \frac{4 \beta (\gamma - 0.5 \beta - 1)}{\mathcal{A}}.  \label{eq:sigma_c}
\end{align}
We now state the following theorem.
\begin{theorem}
    Suppose $\mathcal{R}_0 > 2$ holds. Furthermore, let \eqref{eq:cond1} and \eqref{eq:nc_x^*=0.5} also hold. Then the following conditions are true regarding the stability of the unique endemic SLE $\zt^*_{{\mathtt{EE(int)}}} = (0.5, 1 - \frac{2 \gamma}{\beta})$. If
    \begin{itemize}
        \item $\mathcal{A} \geq 0$, then $\zt^*_{{\mathtt{EE(int)}}}$ is LAS for any $\sigma \geq 0$; and
        \item $\mathcal{A} < 0$, then $\zt^*_{{\mathtt{EE(int)}}}$ is LAS for $\sigma < \sigma^{\mathtt{c}}$ and unstable for $\sigma > \sigma^{\mathtt{c}}$.
    \end{itemize}
    \label{thm:stab_x=0.5}
\end{theorem}
\begin{proof}
    Under the conditions stated in Lemma \ref{lemm:x^*=0.5_exist_unique}, the unique SLE $\zt^*_{\mathtt{EE(int)}}$ exists. We first note that conditions \eqref{eq:nc_x^*=0.5} and \eqref{eq:cond1} are compatible with both $\mathcal{A} \geq 0$ and $\mathcal{A} < 0$. Adding \eqref{eq:A} and \eqref{eq:nc_x^*=0.5} yields $\mathcal{A} = 4\gamma\delta_{TR1} + 2(\beta-2\gamma)\delta_{TR0}$, so that $\mathcal{A} \geq 0$ if and only if $\delta_{TR0} \geq -\frac{2\gamma}{\beta-2\gamma}\delta_{TR1}$. Under \eqref{eq:nc_x^*=0.5}, condition \eqref{eq:cond1} imposes only the lower bound
    \begin{equation*}
    \delta_{TR0} \geq -\frac{2\gamma}{\beta-2\gamma}(\delta_{TR1}+\delta_{PS1}) + \frac{\gamma}{\beta-\gamma}\,\delta_{PS1},
    \end{equation*}
    which is strictly smaller than $-\frac{2\gamma}{\beta-2\gamma}\delta_{TR1}$. Hence, $\delta_{TR0}$ may lie on either side of this value, and both cases of the theorem are admissible. 

    Now, computing the Jacobian at $\zt^*_{\mathtt{EE(int)}} = (0.5, 1 - \frac{2 \gamma}{\beta})$ from equations \eqref{eq:idot_i}-\eqref{eq:xdot_x} yields
    \begin{align}
        &\frac{\partial \dot{\It}(\zt^*_{\mathtt{EE(int)}})}{\partial \It} = \gamma - 0.5 \beta, \quad \frac{\partial \dot{\It}(\zt^*_{\mathtt{EE(int)}})}{\partial x} \! = \! -2 \gamma \bigg(1 - \frac{2 \gamma}{\beta}\bigg), \nonumber
    \\ &\frac{\partial \dot{x}(\zt^*_{\mathtt{EE(int)}})}{\partial \It} = -0.125 \sigma (\delta_{TR0} + \delta_{PS0} - \delta_{TR1} - \delta_{PS1}), \nonumber
    \\ &\frac{\partial \dot{x}(\zt^*_{\mathtt{EE(int)}})}{\partial x} = -0.25 \sigma \bigg(\!\frac{2 \gamma}{\beta} (\delta_{TR1} \! - \! \delta_{PS1}) \nonumber
    \\ &\qquad \qquad \quad \quad + \big(1  - \frac{2 \gamma}{\beta} \big) (\delta_{TR0} - \delta_{PS0}) \bigg) - 1. \nonumber
    \end{align}
    The SLE $(0.5, 1 - \frac{2 \gamma}{\beta})$ is stable if and only if $$\mathtt{Tr}[\mathcal{J}] = \frac{\partial \dot{\It}}{\partial \It} + \frac{\partial \dot{x}}{\partial x} < 0, \textnormal{ and}$$ $$\mathtt{Det}[\mathcal{J}] = \frac{\partial \dot{\It}}{\partial \It} \cdot \frac{\partial \dot{x}}{\partial x} - \frac{\partial \dot{\It}}{\partial x} \cdot \frac{\partial \dot{x}}{\partial \It} > 0$$ simultaneously hold (we have suppressed the dependency on $\zt^*_{\mathtt{EE(int)}}$ for brevity). We first compute $\mathtt{Tr}[\mathcal{J}]$ as
    \begin{align}
        &\mathtt{Tr}[\mathcal{J}] = \gamma - 0.5 \beta -0.25 \sigma \bigg(\!\frac{2 \gamma}{\beta} (\delta_{TR1} \! - \! \delta_{PS1}) \nonumber
    \\ &\qquad \qquad \quad + \! \big(\!1 \! - \! \frac{2 \gamma}{\beta} \! \big) \! (\delta_{TR0} \! - \! \delta_{PS0}) \! \bigg) \! - \! 1 < 0 \nonumber
    \\ \Leftrightarrow &\sigma \bigg(\underbrace{2 \gamma (\delta_{TR1} \! - \! \delta_{PS1}) + \! \big(\beta - 2 \gamma \big) \! (\delta_{TR0} - \delta_{PS0})}_{\mathcal{A}} \bigg) \nonumber
    \\ &\qquad \qquad \quad \; > 4 \beta \underbrace{(\gamma - 0.5 \beta - 1)}_{< 0}.
    \label{eq:trace_cond}
    \end{align}
    When $\mathcal{A} \geq 0$ holds, \eqref{eq:trace_cond} is satisfied for all values of $\sigma \geq 0$.
    %, which is the first item of this theorem (the stability proof is still incomplete when $\mathcal{A} \geq 0$, since we have not yet shown that $\mathtt{Det}[\mathcal{J}(\zt^*)] > 0$). 
    Otherwise, when $\mathcal{A} < 0$ holds, inequality \eqref{eq:trace_cond} reduces to $$\sigma < \frac{4 \beta (\gamma - 0.5 \beta - 1)}{\mathcal{A}} =: \sigma^{\mathtt{c}}.$$
    Thus, for rationality parameters satisfying $\sigma > \sigma^{\mathtt{c}}$ under the constraint of $\mathcal{A} < 0$, the SLE $\zt^*_{\mathtt{EE(int)}}$ is unstable. Finally, we compute $\mathtt{Det}[\mathcal{J}]$ as
    \begin{small}
    \begin{align}
        &\mathtt{Det}[\mathcal{J}] = (\gamma - 0.5 \beta) \big( \frac{-0.25 \sigma \mathcal{A}}{\beta} - 1 \big) - 0.25 \sigma \gamma \nonumber
        \\ &\qquad \qquad \qquad \times \big(\! 1 \! - \! \frac{2 \gamma}{\beta}\! \big) (\delta_{TR0} \! + \! \delta_{PS0} \! - \! \delta_{TR1} \! - \! \delta_{PS1}) > 0 \nonumber
        \\ \Leftrightarrow &0.5 (\underbrace{2 \gamma - \beta}_{< 0 \textnormal{ as } \mathcal{R}_0 > 2})  \big( -0.25 \sigma \mathcal{A} - \beta \big) + 0.25 \sigma \gamma \nonumber
        \\ &\qquad \qquad \qquad \times(2 \gamma - \beta) (\delta_{TR0} \! + \! \delta_{PS0} \! - \! \delta_{TR1} \! - \! \delta_{PS1}) > 0 \nonumber 
        \\ \Leftrightarrow &2  \big( -0.25 \sigma \mathcal{A} - \beta \big) +\sigma \gamma  (\delta_{TR0} \! + \! \delta_{PS0} \! - \! \delta_{TR1} \! - \! \delta_{PS1}) < 0 \nonumber 
        \\ \Leftrightarrow &-2 \beta + \sigma \big(-2 \gamma (\delta_{TR1} - \delta_{TR0}) - 0.5 \beta (\delta_{TR0} - \delta_{PS0}) \big) < 0,
        \label{eq:int_det_cond}
    \end{align}
    \end{small}
    where \eqref{eq:int_det_cond} is obtained by substituting $\mathcal{A}$ from \eqref{eq:A}. Now, computing $-2 \beta - 0.5 \beta \sigma H(0.5)$ by substituting $x^*_{\mathtt{EE(int)}} = 0.5$ in \eqref{eq:H_x} we obtain
    \begin{align}
        -2 \beta - 0.5 \beta \sigma H(0.5) &= -2 \beta + \sigma \big(-2 \gamma (\delta_{TR1} - \delta_{TR0}) \nonumber
        \\ &\quad - 0.5 \beta (\delta_{TR0} - \delta_{PS0}) \big), 
        \label{eq:new_eq_H}
    \end{align}
which is exactly the L.H.S of inequality \eqref{eq:int_det_cond}. Now, in Lemma \ref{lemm:x^*=0.5_exist_unique} we proved that $H(x^*_{\mathtt{EE}}) \geq 0$ along the endemic branch, implying that $H(0.5) \geq 0$ under \eqref{eq:nc_x^*=0.5} and \eqref{eq:cond1}. Therefore, the expression in \eqref{eq:new_eq_H} is strictly negative. Thus, we conclude that inequality \eqref{eq:int_det_cond} is satisfied for all $\sigma$. Consequently, $\mathtt{Det}[\mathcal{J}] > 0$ holds for all $\sigma$. Recall that when $\mathcal{A} < 0$, \eqref{eq:trace_cond} holds if and only if $\sigma < \sigma^{\mathtt{c}}$, hence SLE $\zt^*_{\mathtt{EE(int)}}$ is LAS if $\sigma < \sigma^{\mathtt{c}}$ and unstable for $\sigma > \sigma^{\mathtt{c}}$. At $\sigma = \sigma^{\mathtt{c}}$, the linearization is non-hyperbolic and the theorem makes no stability claim. Furthermore, when $\mathcal{A} \geq 0$ holds, SLE $\zt^*_{\mathtt{EE(int)}}$ is LAS for all $\sigma$. This completes the stability proof.
\end{proof}

Theorem \ref{thm:stab_x=0.5} highlights a very important aspect of logit learning dynamics. When the payoffs and disease parameters are such that $\Delta \ut(0.5, 1 - \frac{2 \gamma}{\beta}) = 0$ holds true at steady-state, then there exists a certain condition (i.e., $\mathcal{A} < 0$) under which the fixed point is unstable when individuals are highly rational. High rationality makes the logit response highly sensitive to small payoff differences, thereby amplifying behavioral responses to perturbations. Thus, a slight increase (or, decrease) in infection level causes a fluctuation in $\Delta \ut$ which in turn drives $x(t)$ away from the equilibrium $x^*_{\mathtt{EE(int)}}=0.5$. However, since agents of lower rationality have less sensitivity to the slight fluctuations in payoffs, $x^*_{\mathtt{EE(int)}} = 0.5$ remains a stable finite-$\sigma$ equilibrium.

\subsection{LAS of $\zt^*_{\mathtt{EE(\sigma)}}$}
\label{sec:LAS_x^*=0}
Finally, we analyze the LAS of the finite-$\sigma$ endemic SLE $\zt^*_{\mathtt{EE(\sigma)}} = (x^*_{\mathtt{EE(\sigma)}}, \It^*_{\mathtt{EE(\sigma)}})$. 
\begin{theorem}
    Let $\mathcal{R}_0$ satisfy $1 < \mathcal{R}_0 \leq 2$, or $\mathcal{R}_0 > 2$ with violation of equality \eqref{eq:nc_x^*=0.5}. Suppose $\infty > \sigma > \max \{0, \sigma^{\mathtt{TH}} \}$, condition \eqref{eq:min_suff_cond_x^*=x^*_sigmaEE} and 
    \begin{align}
        \min \Big\{\delta_{TR1} - \delta_{PS1}, &\gamma (\delta_{TR1} - \delta_{PS1}) \nonumber
        \\ &+ (\beta - \gamma) (\delta_{TR0} -\delta_{PS0})\Big\}\geq 0
        \label{eq:x^*=x^*ee_stability}
    \end{align}
    hold. Then the unique finite-$\sigma$ endemic SLE $\zt^*_{\mathtt{EE(\sigma)}} = (x^*_{\mathtt{EE(\sigma)}}, \It^*_{\mathtt{EE(\sigma)}})$ is LAS.
    \label{thm:x^*=x^*ee_stable}
\end{theorem}
\begin{proof}
    In an endemic equilibrium, \eqref{eq:sis_dyn} reduces to $\beta (1 - x^*_{\mathtt{EE(\sigma)}}) (1 - \It^*_{\mathtt{EE(\sigma)}}) = \gamma$, which we substitute in \eqref{eq:idot_i} to obtain 
    \begin{equation}
        \frac{\partial \dot{\It}(\zt^*_{\mathtt{EE(\sigma)}})}{\partial \It} = -\beta (1 - x^*_{\mathtt{EE(\sigma)}}) \It^*_{\mathtt{EE(\sigma)}}.
        \label{eq:x*ee_idot_i}
    \end{equation}
    From \eqref{eq:diff_payoff}, we obtain 
    $$\frac{\partial \Delta \ut \!(\!\zt^*_{\mathtt{EE(\sigma)}}\!)}{\partial x} \!\!= \!\! (1 \!- \! \It^*_{\mathtt{EE(\sigma)}} \!)(\delta_{TR1} \! - \! \delta_{PS1} \!) \! + \! \It^*_{\mathtt{EE(\sigma)}} \! (\delta_{TR0} \! - \! \delta_{PS0} \!),$$
    which is non-negative under the condition \eqref{eq:x^*=x^*ee_stability}, since $\frac{\partial \Delta \ut (\zt^*_{\mathtt{EE(\sigma)}})}{\partial x}$ is affine in $\It^*_{\mathtt{EE(\sigma)}}$ with endpoints as $\delta_{TR1} - \delta_{PS1}$ and $\frac{1}{\beta} \Big(\gamma (\delta_{TR1} - \delta_{PS1}) + (\beta - \gamma) (\delta_{TR0} -\delta_{PS0}) \Big)$ (the reasoning is similar to the uniqueness proof in Lemma \ref{lemm:x^*=0}, and hence is omitted). Thus, the trace of $\mathcal{J}(\zt^*_{\mathtt{EE(\sigma)}})$ is
    \begin{align*}
        \mathtt{Tr}[\mathcal{J}] &= -\beta (1 - x^*_{\mathtt{EE(\sigma)}}) \It^*_{\mathtt{EE(\sigma)}} 
        \\ &\quad- \frac{\sigma e^{\sigma \Delta \ut (\zt^*_{\mathtt{EE(\sigma)}})}}{(1 + e^{\sigma \Delta \ut(\zt^*_{\mathtt{EE(\sigma)}})})^2} \cdot \frac{\partial \Delta \ut (\zt^*_{\mathtt{EE(\sigma)}})}{\partial x} - 1,    
    \end{align*}
    which is strictly negative.

    We now compute the determinant as
    \begin{small}
    \begin{align}
        \mathtt{Det}[\mathcal{J}] &= \beta (1 - x^*_{\mathtt{EE(\sigma)}}) \It^*_{\mathtt{EE(\sigma)}} \bigg(1 + \frac{\sigma e^{\sigma \Delta \ut (\zt^*_{\mathtt{EE(\sigma)}})}}{(1 + e^{\sigma \Delta \ut (\zt^*_{\mathtt{EE(\sigma)}})})^2} \nonumber
        \\ &\quad \times\frac{\partial \Delta \ut(\zt^*_{\mathtt{EE(\sigma)}})}{\partial x} \bigg) - \beta \It^*_{\mathtt{EE(\sigma)}} (1 - \It^*_{\mathtt{EE(\sigma)}}) \nonumber
        \\ &\quad \times\frac{\sigma e^{\sigma \Delta \ut (\zt^*_{\mathtt{EE(\sigma)}})}}{(1 + e^{\sigma \Delta \ut (\zt^*_{\mathtt{EE(\sigma)}})})^2}  \cdot \frac{\partial \Delta \ut(\zt^*_{\mathtt{EE(\sigma)}})}{\partial\It} \nonumber 
        \\ &= \beta (1 - x^*_{\mathtt{EE(\sigma)}}) \It^*_{\mathtt{EE(\sigma)}} \bigg(1 + \frac{\sigma e^{\sigma \Delta \ut (\zt^*_{\mathtt{EE(\sigma)}})}}{(1 + e^{\sigma \Delta \ut (\zt^*_{\mathtt{EE(\sigma)}})})^2} \nonumber
        \\ &\quad \times \big(\frac{\partial \Delta \ut(\zt^*_{\mathtt{EE(\sigma)}})}{\partial x} - \frac{\gamma}{\beta (1 - x^*_{\mathtt{EE(\sigma)}})^2} \cdot \frac{\partial \Delta \ut(\zt^*_{\mathtt{EE(\sigma)}})}{\partial\It} \big) \bigg),
        \label{eq:det_eq_1}
    \end{align}
    \end{small}
    where the last equality is obtained by substituting $1 - \It^*_{\mathtt{EE(\sigma)}} = \frac{\gamma}{\beta (1 - x^*_{\mathtt{EE(\sigma)}})}$. For any $x$ along the endemic branch, we expand $\frac{d \Delta \ut (\zt^*_{\mathtt{EE(\sigma)}})}{d x}$ using chain rule to obtain
    \begin{small}
    \begin{align}
        \frac{d \Delta \ut (\zt^*_{\mathtt{EE(\sigma)}})}{d x} &= \frac{\partial \Delta \ut (\zt^*_{\mathtt{EE(\sigma)}})}{\partial x} \nonumber
        \\ &\quad + \frac{\partial \Delta \ut (\zt^*_{\mathtt{EE(\sigma)}})}{\partial \It} \cdot \frac{d }{d x} \big(1 - \frac{\gamma}{\beta (1 - x)}\big) \Bigg|_{x = x^*_{\mathtt{EE(\sigma)}}} \nonumber
        \\ &= \frac{\partial \Delta \ut(\zt^*_{\mathtt{EE(\sigma)}})}{\partial x} - \frac{\gamma}{\beta (1 - x^*_{\mathtt{EE(\sigma)}})^2} \cdot \frac{\partial \Delta \ut(\zt^*_{\mathtt{EE(\sigma)}})}{\partial\It}. \nonumber
    \end{align}
    \end{small}
    Substituting above equality in \eqref{eq:det_eq_1} simplifies to
    \begin{small}
    $$\mathtt{Det}[\mathcal{J}] \!\! = \!\! \beta (\!1 \! - \! x^*_{\mathtt{EE(\sigma)}}) \It^*_{\mathtt{EE(\sigma)}} \bigg(\!\!1 \!+ \frac{\sigma e^{\sigma \Delta \ut (\zt^*_{\mathtt{EE(\sigma)}})}}{(1 \! + \! e^{\sigma \Delta \ut (\zt^*_{\mathtt{EE(\sigma)}})})^2} \cdot \frac{d \Delta \ut (\zt^*_{\mathtt{EE(\sigma)}})}{d x} \!\bigg),$$
    \end{small}
    which is strictly positive, since $\frac{d \Delta \ut (\zt^*_{\mathtt{EE(\sigma)}})}{d x} \geq 0$ under the uniqueness condition \eqref{eq:min_suff_cond_x^*=x^*_sigmaEE} (see proof of Lemma \ref{lemm:x^*=x^*_ee}). Thus, the finite-$\sigma$ endemic SLE $\zt^*_{\mathtt{EE(\sigma)}}$ is LAS.
\end{proof}

LAS of corner endemic SLE $\zt^*_{\mathtt{EE(cor)}} = (0, 1 - \frac{\gamma}{\beta})$ follows as a limiting case when $\sigma \rightarrow \infty$. To this end, we state the following corollary.
\begin{corollary}
    Suppose $\mathcal{R}_0 > 1$ and inequality \eqref{eq:delta2} hold. Then in the fully rational case of $\sigma \rightarrow \infty$, the corner endemic SLE $\zt^*_{\mathtt{EE(cor)}} = (0, 1 - \frac{\gamma}{\beta})$ is LAS. If condition \eqref{eq:min_suff_cond_x^*=0} is also satisfied, then $\zt^*_{\mathtt{EE(cor)}}$ is the unique SLE that is LAS.
    \label{thm:x^*=0_stable}
\end{corollary}
\begin{proof}
    The proof is presented in the Appendix.
\end{proof}

As previously discussed in the characterization \eqref{eq:sigma<<infyFPP} in Proposition \ref{prop:finite_sigma_x^*}, at any $\sigma < \infty$ the equilibrium $x^*_{\mathtt{EE(cor)}} = 0$ no longer exists and $x^*_{\mathtt{EE(\sigma)}} > 0$ emerges as the finite-$\sigma$ equilibrium strategy.
Note that $x^*_{\mathtt{EE(cor)}} = 0$ implies that all individuals defect, and the infection is at its highest level. Since this is not desirable, we derive an analytical expression for $\sigma$, such that the cooperation fraction reaches a suitable value at steady-state of $0 < x^*_{\mathtt{EE(\sigma)}} = \epsilon < 1 - \frac{\gamma}{\beta}$, and the infection reduces to $\It^*_{\mathtt{EE(\sigma)}} = 1 - \frac{\gamma}{\beta (1 - \epsilon)} < 1 - \frac{\gamma}{\beta}$. Substituting $x^*_{\mathtt{EE(\sigma)}} = \epsilon$ and $\It^*_{\mathtt{EE(\sigma)}} = 1 - \frac{\gamma}{\beta (1 - \epsilon)}$ in \eqref{eq:diff_payoff}, we obtain 
\begin{align}
    \Delta \ut&(x^*_{\mathtt{EE(\sigma)}} = \epsilon, \It^*_{\mathtt{EE(\sigma)}} = 1 - \frac{\gamma}{\beta (1 - \epsilon)}) \nonumber
    \\&= \frac{\gamma}{\beta (1 - \epsilon)} \times \big(\epsilon \delta_{TR1} + (1 - \epsilon) \delta_{PS1} \big) \nonumber
    \\ &\quad +\big(1 - \frac{\gamma}{\beta (1 - \epsilon)} \big) \times \big(\epsilon \delta_{TR0} + (1 - \epsilon) \delta_{PS0}\big).
    \label{eq:x^*=epsilon}
\end{align}
Substituting \eqref{eq:x^*=epsilon} and $x^*_{\mathtt{EE(\sigma)}} = \epsilon$ in \eqref{eq:logit_eq_pt} gives the expression of rationality parameter as 
\begin{equation}
    \sigma \! = \! \frac{1}{\Delta \ut(x^*_{\mathtt{EE(\sigma)}} = \epsilon, \It^*_{\mathtt{EE(\sigma)}} = 1 - \frac{\gamma}{\beta (1 - \epsilon)})} \! \! \times \mathtt{ln} \big(\frac{1}{\epsilon} - 1 \big) =: \sigma^{\epsilon}.
    \label{eq:sigma^th_2}
\end{equation}
It is interesting to note that a change in the rationality parameter causes the fraction of the population that cooperates to change, resulting in a reduced infection level. A value of $\sigma =\sigma^{\epsilon}$ ensures that $(x^*_{\mathtt{EE(\sigma)}} = \epsilon, \It^*_{\mathtt{EE(\sigma)}} = 1 - \frac{\gamma}{\beta (1 - \epsilon)})$ is the finite-$\sigma$ SLE. An increase in $\sigma$ reduces $x^*_{\mathtt{EE(\sigma)}}$, and in the limiting case of $\sigma \rightarrow \infty$ the stable SLE that exists is $\zt^*_{\mathtt{EE(cor)}}$.

\begin{remark}
    We briefly highlight a case under $\sigma \rightarrow \infty$, such that, no endemic SLE exists. Consider that the equality
    \eqref{eq:nc_x^*=0.5} does not hold, i.e.,
    \[
        2 \gamma (\delta_{TR1} + \delta_{PS1})
        + (\beta - 2 \gamma) (\delta_{TR0} + \delta_{PS0}) \neq 0,
    \]
    and the payoffs satisfy
    \begin{align*}
        \Delta \ut\left(0,1-\frac{\gamma}{\beta}\right) &= \frac{\gamma}{\beta}(\delta_{PS1}-\delta_{PS0})
        +\delta_{PS0}<0,
        \\ \Delta \ut(1,0) &=\delta_{TR1}>0.
    \end{align*}
    Hence, neither boundary strategy is self-consistent: $x^*=0$ requires $\Delta \ut(x^*=0,\cdot)>0$, whereas $x^*=1$ requires $\Delta \ut(x^*=1,\cdot)<0$. Moreover, under complete rationality, an interior fixed point requires $x^*=0.5$ and $\Delta \ut(x^*,1-\frac{\gamma}{\beta(1-x^*)})=0$ to hold simultaneously. Since condition \eqref{eq:nc_x^*=0.5} does not hold, the latter condition is not satisfied at $x^*=0.5$. Therefore, no endemic SLE exists in the complete-rationality limit under these sets of conditions. This conclusion is specific to the pointwise limit convention, by Lemma \ref{lemm:x^*=x^*_ee}, an endemic SLE exists for every sufficiently large finite $\sigma$, and these SLE accumulate at states with $\Delta \ut = 0$ and $x^* \neq 0.5$, which are best-response equilibria but not SLE of the pointwise limit.
    \label{rem:del_U=0_unstable}
\end{remark}

In the following section, we validate some of our main analytical results using numerical simulations.

\section{Simulation Results}
\label{sec:num_sim}

In this section, we choose a set of parameters that satisfies $1 < \mathcal{R}_0 < 2$, and another set such that $\mathcal{R}_0 > 2$ is true. 
\begin{table}[htbp]
\caption{}
  \label{tab:table1}
  \centering
  \resizebox{8.6cm}{!}{
  \begin{tabular}{|c | c | c | c | c | c | c | c |} 
   \hline
   $\beta$ & $\gamma$ & $\It(0)$ & $x(0)$ & $\delta_{PS1}$ & $\delta_{TR1}$ & $\delta_{PS0}$ & $\delta_{TR0}$\\ [0.5ex] 
   \hline
   0.75 & 0.45 & 0.1 & 0.5 & 0.2 & 0.3 & -0.2 & 0.1\\ \hline
  \end{tabular}
  }
\end{table}
Note that the parameters in Table \ref{tab:table1} satisfy $\mathcal{R}_0 > 1$ (Lemmas \ref{lemm:x^*=x^*_ee} and \ref{lemm:x^*=0} and Corollary \ref{thm:x^*=0_stable}); $1 < \mathcal{R}_0 < 2$ (Theorem \ref{thm:dfe_stab}); Lemma \ref{lemm:x^*=x^*_ee} inequality \eqref{eq:min_suff_cond_x^*=x^*_sigmaEE}; Lemma \ref{lemm:x^*=0} inequalities \eqref{eq:delta2} and \eqref{eq:min_suff_cond_x^*=0}; $\delta_{TR1} \!\! \geq \!\! \delta_{PS1}$ (Theorem \ref{thm:dfe_stab}); Theorem \ref{thm:x^*=x^*ee_stable} inequality \eqref{eq:x^*=x^*ee_stability}.

First, we illustrate the LAS of the disease-free SLE $\zt^*_{\mathtt{DFE}}$ (Theorem \ref{thm:dfe_stab}) and finite-$\sigma$ endemic SLE $\zt^*_{\mathtt{EE(\sigma)}}$ (Theorem \ref{thm:x^*=x^*ee_stable}). For the selected set of parameters in Table \ref{tab:table1}, the threshold rationality $\sigma^{\mathtt{TH}}$ computed from \eqref{eq:sigma^th_1} in Theorem \ref{thm:dfe_stab} is $1.6894$. We consider two cases in Figure \ref{fig:sigma^th_2}: a) $\sigma =  \sigma^{\mathtt{TH}} + 3 = 4.6894$, and b) $\sigma =  \sigma^{\mathtt{TH}} - 1 = 0.6894$.

\begin{figure}[h!]
\centering
  \subfigure{\includegraphics[width=40mm]{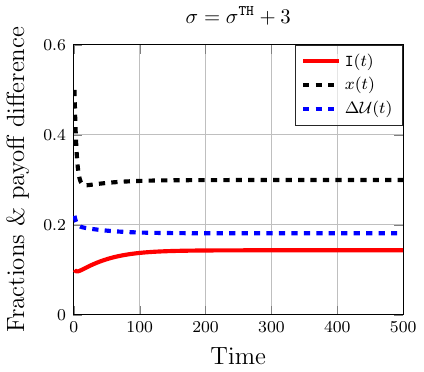}}
  \hspace{3mm}
  % \subfigure{\includegraphics[width=56mm]{figures/sigma2.pdf}}
  % \hspace{3mm}
  \subfigure{\includegraphics[width=40mm]{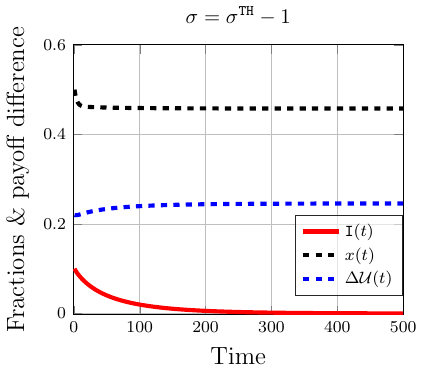}}
  \caption{Evolution of the infected fraction $\It(t)$, cooperation fraction $x(t)$, and payoff difference $\Delta \ut(t)$ (when $\sigma^{\mathtt{TH}} = 1.6894$ in Theorem \ref{thm:dfe_stab} holds) under $\sigma$ values of $4.6894$ (left), and $0.6894$ (right).}
  \label{fig:sigma^th_2}
\end{figure}
When $\sigma > \sigma^{\mathtt{TH}}$ holds, we observe that a stable finite-$\sigma$ endemic equilibrium $\zt^*_{\mathtt{EE(\sigma)}}$ exists (left plot of Figure \ref{fig:sigma^th_2}). This demonstrates that disease-free SLE is (locally) unstable. In the second case where $\sigma < \sigma^{\mathtt{TH}}$, the right plot of Figure \ref{fig:sigma^th_2} illustrates the LAS of the disease-free SLE $\zt^*_{\mathtt{DFE}}$. The above observations are consistent with the theoretical analysis in Theorems \ref{thm:dfe_stab} and \ref{thm:x^*=x^*ee_stable}.
\begin{figure}[h!]
\centering
  \subfigure{\includegraphics[width=28mm]{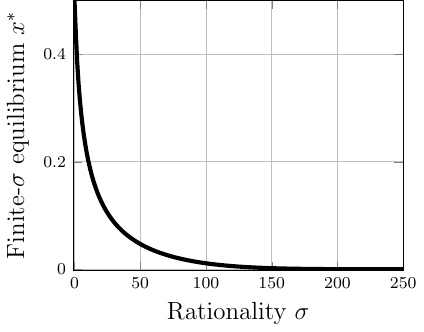}}
  %\hspace{1mm}
  \subfigure{\includegraphics[width=28mm]{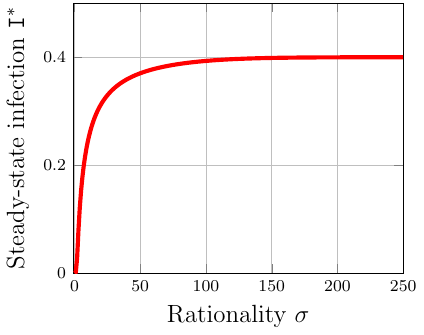}}
  %\hspace{1mm}
  \subfigure{\includegraphics[width=28mm]{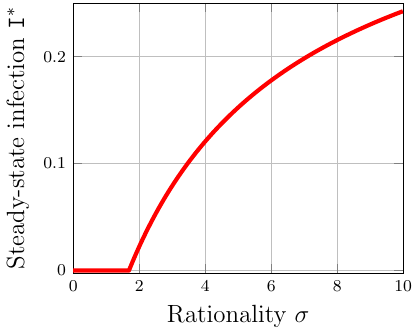}}
  \caption{Variation of fixed points $x^*$ (left), and $\It^*$ (middle) with $\sigma$. Variation of $\It^*$ (right) with $\sigma$ on a magnified scale.}
  \label{fig:var_sigma}
\end{figure}

\begin{figure*}[h!]
\centering
  \subfigure{\includegraphics[width=55mm]{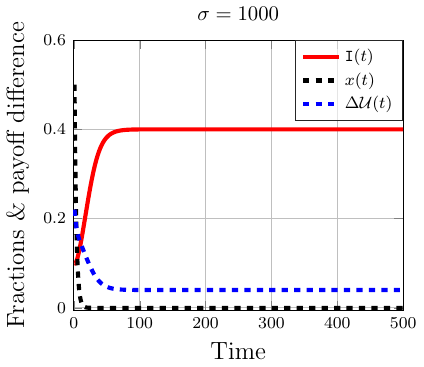}}
  \hspace{3mm}
  \subfigure{\includegraphics[width=56mm]{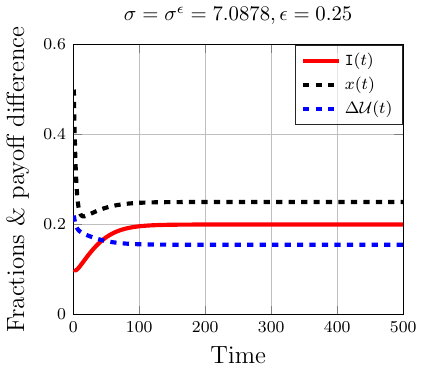}}
  \hspace{3mm}
  \subfigure{\includegraphics[width=56mm]{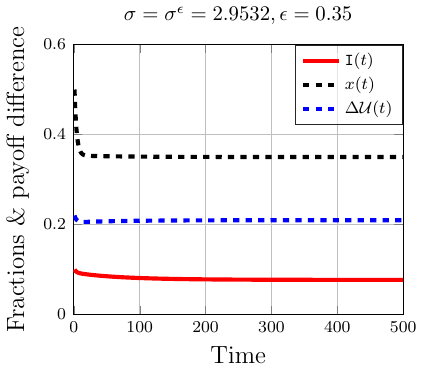}}
  \caption{Evolution of the infected fraction $\It(t)$, cooperation fraction $x(t)$, and payoff difference $\Delta \ut(t)$ under $\sigma$ values of $1000$ (left), $7.0878$ (middle), and $2.9532$ (right), when the parameters satisfy $R_1 = 0.3, S_1 = 0.2, T_1 = 0.6, P_1 = 0.4, R_0 = 0.5, S_0 = 0.4, T_0 = 0.6$ and $P_0 = 0.2$.}
  \label{fig:sigma^th_1}
\end{figure*}

\begin{figure*}[h!]
\centering
  \subfigure{\includegraphics[width=55mm]{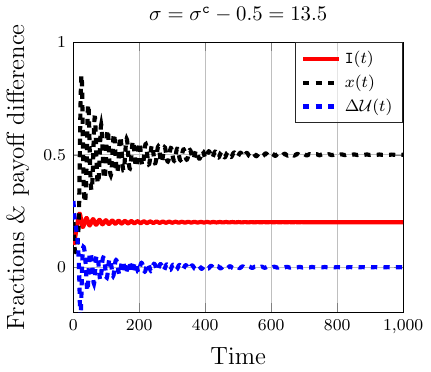}}
  \hspace{3mm}
  \subfigure{\includegraphics[width=56mm]{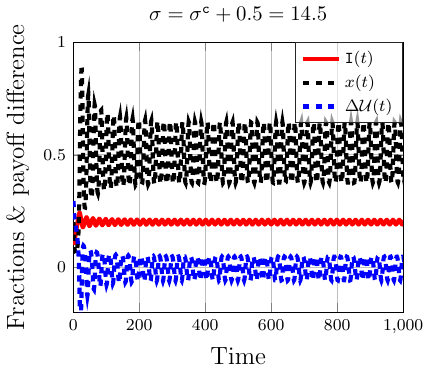}}
  \hspace{3mm}
  \subfigure{\includegraphics[width=56mm]{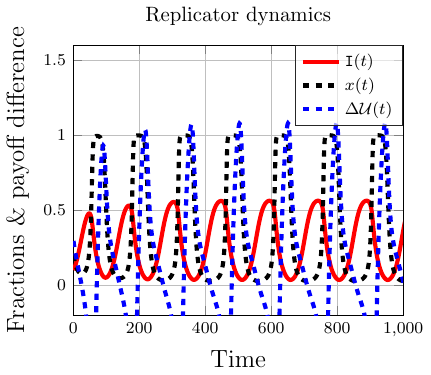}}
  \caption{Comparison of logit dynamics under $\sigma$ values of $13.5$ (left), and $14.5$ (middle), and replicator dynamics (right), when the parameters satisfy $R_1 = 1, S_1 = 1, T_1 = 2.5, P_1 = 1.25, R_0 = 8, S_0 = 2, T_0 = 1.25$ and $P_0 = 1.75$.}
  \label{fig:sigma+repl}
\end{figure*}

To illustrate the steady decrease (increase) of $x^*$ ($\It^*$, respectively) with $\sigma$ for the selected parameter set, we vary the rationality parameter from $0$ to $250$, and plot the SLE for each value of sigma. Figure \ref{fig:var_sigma} studies this variation. As expected, we observe that at $\sigma = 0$, exactly half of the population cooperates, leading to the disease-free equilibrium with $\It^*_{\mathtt{DFE}} = 0$. As rationality increases, the fraction of population cooperating falls, which in turn increases the infection level. Finally, at a very high rationality we observe that $x^*$ is close to zero and the steady-state infection is the highest. Thus, the transition is from SLE $\zt^*_{\mathtt{DFE}}$ to SLE $\zt^*_{\mathtt{EE(\sigma)}}$. At higher values of rationality, the finite-$\sigma$ SLE $\zt^*_{\mathtt{EE(\sigma)}}$ approaches the corner endemic SLE $\zt^*_{\mathtt{EE(cor)}}$.

To realize the transition from disease-free equilibrium to endemic equilibrium, we reduce the scale of $\sigma$ and magnify the plot in the middle panel. In the right plot of Figure \ref{fig:var_sigma}, we observe that the infected proportion is zero at $\sigma < 1.6894 = \sigma^{\mathtt{TH}}$. Beyond the threshold we find that steady-state infection is positive and disease-free SLE $\zt^*_{\mathtt{DFE}}$ becomes unstable (Theorems \ref{thm:dfe_stab} and \ref{thm:x^*=x^*ee_stable}).

In Figure \ref{fig:sigma^th_1}, we illustrate the trajectories of infection level $\It(t)$, proportion of population cooperating $x(t)$, and payoff difference $\Delta \ut (t)$ for different levels of rationality. Under the setting of high rationality, we show numerically in the left plot of Figure \ref{fig:sigma^th_1} that in the selected parameter regime, the trajectories converge to a stable finite-$\sigma$ SLE close to $\zt^*_{\mathtt{EE(cor)}} = (0, 1 - \frac{\gamma}{\beta})$. Since $\sigma = 1000$ is finite, SLE $\zt^*_{\mathtt{EE(\sigma)}}$ is not exactly $(0, 0.4)$, but in the limiting case of $\sigma \rightarrow \infty$, the stable SLE is $(0, 0.4)$ (Lemmas \ref{lemm:x^*=0} and \ref{lemm:x^*=x^*_ee}, Theorem 
\ref{thm:x^*=x^*ee_stable} and Corollary \ref{thm:x^*=0_stable}).

We further obtain the values of rationality from \eqref{eq:sigma^th_2} as $\sigma^{\epsilon}\bigg|_{\epsilon = 0.25} = 7.0878$ and $\sigma^{\epsilon}\bigg|_{\epsilon = 0.35} = 2.9532$ by setting $x^*_{\mathtt{EE(\sigma)}} = \epsilon = 0.25$ and $x^*_{\mathtt{EE(\sigma)}} = \epsilon = 0.35$, respectively. Evolution of the infected proportion and strategy for different values of $\epsilon$ are illustrated in the middle and right panels of Figure \ref{fig:sigma^th_1}. The value $\sigma^{\epsilon}$ is the rationality level for which the prescribed cooperation fraction $\epsilon$ constitutes the finite-$\sigma$ SLE $(\epsilon, 1 - \frac{\gamma}{\beta (1 - \epsilon)})$.

To illustrate the existence and LAS of the interior SLE $\zt^*_{\mathtt{EE(int)}}$, we choose the parameters according to Table \ref{tab:table2} for the next set of simulations.
\vspace{-1.5mm}
\begin{table}[htbp]
\caption{}
  \label{tab:table2}
  \centering
\resizebox{8.6cm}{!}{
\begin{tabular}{|c | c | c | c | c | c | c | c |} 
 \hline
 $\beta$ & $\gamma$ & $\It(0)$ & $x(0)$ & $\delta_{PS1}$ & $\delta_{TR1}$ & $\delta_{PS0}$ & $\delta_{TR0}$\\ [0.5ex] 
 \hline
 0.5 & 0.2 & 0.1 & 0.5 & 0.25 & 1.5 & -0.25 & -6.75\\ \hline
\end{tabular}
}
\end{table}
These chosen parameters satisfy the conditions \eqref{eq:nc_x^*=0.5} and \eqref{eq:cond1} in Lemma \ref{lemm:x^*=0.5_exist_unique}, and the inequalities $\mathcal{R}_{0} > 2$ (Lemma \ref{lemm:x^*=0.5_exist_unique} and Theorem \ref{thm:stab_x=0.5}), and $\mathcal{A} = -0.15 <  0$ (Theorem \ref{thm:stab_x=0.5}). The value of $\sigma^{\mathtt{c}}$ as derived in Theorem \ref{thm:stab_x=0.5} (equation \eqref{eq:sigma_c}) is $14$.
 
In Figure \ref{fig:sigma+repl}, we not only validate Lemma \ref{lemm:x^*=0.5_exist_unique} and Theorem \ref{thm:stab_x=0.5}, but also compare the trajectories when the learning protocol is governed by replicator dynamics. In the left plot of Figure \ref{fig:sigma+repl}, the value of rationality is lower than the threshold $\sigma^{\mathtt{c}}$. The  plot shows that the unique SLE $\zt^*_{\mathtt{EE(int)}} = (0.5, 0.2)$ is stable. When $\sigma$ exceeds the threshold (middle plot of Figure \ref{fig:sigma+repl}), we observe that $\zt^*_{\mathtt{EE(int)}}$ is no longer stable. It is interesting to note that the average values of $x(t)$ and $\Delta \ut(t)$ are $0.5$ and $0$, respectively. However, a high value of $\sigma$ results in the interior endemic being unstable. We also compare the trajectories and nature of the equilibrium when the strategy revision is according to replicator dynamics. Observe in the right plot of Figure \ref{fig:sigma+repl} that the trajectories do not converge to any steady-state equilibrium. Instead, the trajectories exhibit oscillations over the simulated time horizon. The average values of $x(t)$ and $\It(t)$ beyond $600$ units of time are around $0.39$ and $0.27$. Since the average fraction that cooperates decreases from $0.5$, we notice an increase in the average infection level when compared to logit dynamics. Thus, for the chosen parameter regime, bounded rationality stabilizes the equilibrium and yields a lower time-averaged infected proportion than the replicator dynamics.

\vspace{-0.4mm}
\section{Conclusions and Future Work}
\label{sec:conc}

In this paper we considered an SIS model in which the epidemic propagation and population behavior are coupled. We modeled individual rationality using logit dynamics, and characterized the resulting stationary logit equilibria. We identified certain parameter regimes in which lower rationality leads to a reduced infection level. We showed that in the regime of $1 < \mathcal{R}_0 < 2$, the disease-free equilibrium is unstable under replicator and high rationality logit dynamics but is LAS under low rationality logit learning. Through our numerical simulations we also illustrated that in certain parameter regimes, bounded rationality stabilizes the existing endemic equilibrium and converges to a lower infection level. This convergence is not always present when logit choice is replaced by replicator dynamics. 

While we have derived several conditions that guarantee the existence, uniqueness and LAS of the SLE, this paper does not present a complete characterization of the equilibria, and exploring the remaining parameter spaces remain an open challenge for future research. 
In this work, we have illustrated one example in which trajectories under replicator dynamics fail to converge, whereas, those under bounded rationality converge. A complete theoretical analysis on this comparison remains an interesting problem. 
Extension to other classes of epidemiological models without re-infection risk, and networked settings remain promising directions for future research.

\bibliographystyle{IEEEtran}
\bibliography{ref}

@article{kabir2020evolutionary,
  title={Evolutionary game theory modelling to represent the behavioural dynamics of economic shutdowns and shield immunity in the {COVID}-19 pandemic},
  author={Kabir, KM Ariful and Tanimoto, Jun},
  journal={Royal Society open science},
  volume={7},
  number={9},
  pages={201095},
  year={2020}
}

@article{ye2021game,
  title={Game-theoretic modeling of collective decision making during epidemics},
  author={Ye, Mengbin and Zino, Lorenzo and Rizzo, Alessandro and Cao, Ming},
  journal={Physical Review E},
  volume={104},
  number={2},
  pages={024314},
  year={2021},
  publisher={APS}
}

@article{chang2020game,
  title={Game theoretic modelling of infectious disease dynamics and intervention methods: a review},
  author={Chang, Sheryl L and Piraveenan, Mahendra and Pattison, Philippa and Prokopenko, Mikhail},
  journal={Journal of biological dynamics},
  volume={14},
  number={1},
  pages={57--89},
  year={2020},
  publisher={Taylor \& Francis}
}

@article{huang2022game,
  title={Game-theoretic frameworks for epidemic spreading and human decision-making: {A} review},
  author={Huang, Yunhan and Zhu, Quanyan},
  journal={Dynamic Games and Applications},
  volume={12},
  number={1},
  pages={7--48},
  year={2022},
  publisher={Springer}
}

@inproceedings{paarporn2023sis,
  title={S{IS} epidemics coupled with evolutionary social distancing dynamics},
  author={Paarporn, Keith and Eksin, Ceyhun},
  booktitle={2023 American Control Conference (ACC)},
  pages={4308--4313},
  year={2023},
  organization={IEEE}
}

@inproceedings{watkins2017inference,
  title={Inference, prediction and control of networked epidemics},
  author={Watkins, Nicholas J and Nowzari, Cameron and Pappas, George J},
  booktitle={2017 American Control Conference (ACC)},
  pages={5611--5616},
  year={2017},
  organization={IEEE}
}

@inproceedings{kohler2018dynamic,
  title={Dynamic resource allocation to control epidemic outbreaks a model predictive control approach},
  author={K{\"o}hler, Johannes and Enyioha, Chinwendu and Allg{\"o}wer, Frank},
  booktitle={2018 Annual American Control Conference (ACC)},
  pages={1546--1551},
  year={2018},
  organization={IEEE}
}

@article{watkins2019robust,
  title={Robust economic model predictive control of continuous-time epidemic processes},
  author={Watkins, Nicholas J and Nowzari, Cameron and Pappas, George J},
  journal={IEEE Transactions on Automatic Control},
  volume={65},
  number={3},
  pages={1116--1131},
  year={2019},
  publisher={IEEE}
}

@article{nowzari2015optimal,
  title={Optimal resource allocation for control of networked epidemic models},
  author={Nowzari, Cameron and Preciado, Victor M and Pappas, George J},
  journal={IEEE Transactions on Control of Network Systems},
  volume={4},
  number={2},
  pages={159--169},
  year={2015},
  publisher={IEEE}
}

@article{wei2023adoption,
  title={On the adoption of nonpharmaceutical interventions during the pandemic: An evolutionary game model},
  author={Wei, Zhiyuan and Zhuang, Jun},
  journal={Risk Analysis},
  volume={43},
  number={11},
  pages={2298--2311},
  year={2023},
  publisher={Wiley Online Library}
}

@article{maitra2023sis,
  title={S{IS} epidemic propagation under strategic non-myopic protection: {A} dynamic population game approach},
  author={Maitra, Urmee and Hota, Ashish R and Srivastava, Vaibhav},
  journal={IEEE Control Systems Letters},
  volume={7},
  pages={1578--1583},
  year={2023},
  publisher={IEEE}
}

@article{hota2019game,
  title={Game-theoretic vaccination against networked {SIS} epidemics and impacts of human decision-making},
  author={Hota, Ashish R and Sundaram, Shreyas},
  journal={IEEE Transactions on Control of Network Systems},
  volume={6},
  number={4},
  pages={1461--1472},
  year={2019},
  publisher={IEEE}
}

@inproceedings{pare2015stability,
  title={Stability analysis and control of virus spread over time-varying networks},
  author={Par{\'e}, Philip E and Beck, Carolyn L and Nedi{\'c}, Angelia},
  booktitle={2015 54th IEEE Conference on Decision and Control (CDC)},
  pages={3554--3559},
  year={2015},
  organization={IEEE}
}

@article{hota2023learning,
  title={Learning to mitigate epidemic risks: A dynamic population game approach},
  author={Hota, Ashish R and Maitra, Urmee and Elokda, Ezzat and Bolognani, Saverio},
  journal={Dynamic Games and Applications},
  volume={13},
  number={4},
  pages={1106--1129},
  year={2023},
  publisher={Springer}
}

@article{reluga2010game,
  title={Game theory of social distancing in response to an epidemic},
  author={Reluga, Timothy C},
  journal={PLoS computational biology},
  volume={6},
  number={5},
  pages={e1000793},
  year={2010},
  publisher={Public Library of Science San Francisco, USA}
}

@article{liu2022herd,
  title={Herd behaviors in epidemics: A dynamics-coupled evolutionary games approach},
  author={Liu, Shutian and Zhao, Yuhan and Zhu, Quanyan},
  journal={Dynamic Games and Applications},
  volume={12},
  number={1},
  pages={183--213},
  year={2022},
  publisher={Springer}
}

@book{sandholm2010population,
  title={Population games and evolutionary dynamics},
  author={Sandholm, William H},
  volume={25},
  year={2010},
  publisher={MIT press Cambridge}
}

@article{hofbauer2003evolutionary,
  title={Evolutionary game dynamics},
  author={Hofbauer, Josef and Sigmund, Karl},
  journal={Bulletin of the American mathematical society},
  volume={40},
  number={4},
  pages={479--519},
  year={2003}
}

@article{hofbauer2007evolution,
  title={Evolution in games with randomly disturbed payoffs},
  author={Hofbauer, Josef and Sandholm, William H},
  journal={Journal of economic theory},
  volume={132},
  number={1},
  pages={47--69},
  year={2007},
  publisher={Elsevier}
}

@INPROCEEDINGS{9992844,
  author={Olmez, S. Yagiz and Aggarwal, Shubham and Kim, Jin Won and Miehling, Erik and Başar, Tamer and West, Matthew and Mehta, Prashant G.},
  booktitle={2022 IEEE 61st Conference on Decision and Control (CDC)}, 
  title={How does a Rational Agent Act in an Epidemic?}, 
  year={2022},
  volume={},
  number={},
  pages={5536-5543},
  doi={10.1109/CDC51059.2022.9992844}}

@article{park2025epidemic,
  title={Epidemic population games and perturbed best response dynamics},
  author={Park, Shinkyu and Certorio, Jair and Martins, Nuno C and La, Richard J},
  journal={IEEE Transactions on Automatic Control},
  volume={70},
  number={9},
  pages={6036--6049},
  year={2025},
  publisher={IEEE}
}

@book{blanchini2008set,
  title={Set-theoretic Methods in Control},
  author={Blanchini, Franco and Miani, Stefano},
  edition={2},
  year={2015},
  publisher={Springer}
}

@ARTICLE{11499399,
  author={Gavin, R. and Paarporn, K. and Cao, M.},
  journal={IEEE Transactions on Automatic Control}, 
  title={Learning to Avoid Tragedies: Boundedly Rational Agents in Feedback-Evolving Games}, 
  year={2026},
  volume={},
  number={},
  pages={1-15},
  doi={10.1109/TAC.2026.3685548}}

@article{weitz2016oscillating,
  title={An oscillating tragedy of the commons in replicator dynamics with game-environment feedback},
  author={Weitz, Joshua S and Eksin, Ceyhun and Paarporn, Keith and Brown, Sam P and Ratcliff, William C},
  journal={Proceedings of the National Academy of Sciences},
  volume={113},
  number={47},
  pages={E7518--E7525},
  year={2016},
  publisher={National Academy of Sciences}
}

@article{fu2010imitation,
  title={Imitation dynamics of vaccination behaviour on social networks},
  author={Fu, Feng and Rosenbloom, Daniel I and Wang, Long and Nowak, Martin A},
  journal={Proceedings of the Royal Society B: Biological Sciences},
  volume={278},
  number={1702},
  pages={42},
  year={2010}
}

@inproceedings{satapathi2022epidemic,
  title={Epidemic propagation under evolutionary behavioral dynamics: Stability and bifurcation analysis},
  author={Satapathi, Abhisek and Dhar, Narendra Kumar and Hota, Ashish R and Srivastava, Vaibhav},
  booktitle={2022 American Control Conference (ACC)},
  pages={3662--3667},
  year={2022},
  organization={IEEE}
}

@ARTICLE{10933512,
  author={Cianfanelli, Leonardo and Como, Giacomo},
  journal={IEEE Transactions on Automatic Control}, 
  title={On the Stability of the Logit Dynamics in Population Games}, 
  year={2025},
  volume={70},
  number={9},
  pages={5910-5925},
  doi={10.1109/TAC.2025.3553096}}

@article{maitra2026discrete,
  title={Discrete-Time Sequential Persuasive Game for Containing {SIS} Epidemic Propagation},
  author={Maitra, Urmee and Hota, Ashish Ranjan},
  journal={Decision Analysis},
  note={Article in Advance},
  year={2026},
  doi={https://doi.org/10.1287/deca.2025.0496}
}

\appendix

\underline{\textbf{Proof of Lemma \ref{lemm:pos_inv}:}}

\begin{proof}
    To prove the positive invariance of $\Upgamma$ for any $(x(0), \It(0)) \in [0, 1]^2$, we analyze the states at the boundary faces. At any point $t \geq 0$, if $\It(t) = 0$, then $\dot{\It}(t) = 0$, and if $\It(t) = 1$, then $\dot{\It}(t) = -\gamma < 0$ holds. Similarly, $x(t) = 0$ implies $\dot{x}(t) = \frac{1}{1 + e^{\sigma \Delta \ut(t)}} \geq 0, \forall \sigma$, whereas, when $x(t) = 1$ holds, we obtain $\dot{x}(t) = \frac{1}{1 + e^{\sigma \Delta \ut(t)}} - 1 \leq 0, \forall \sigma$. For every finite $\sigma \geq 0$, the vector field does not point outward on the boundary of $\Upgamma$, and hence $\Upgamma$ is positively invariant by Nagumo's theorem \cite{blanchini2008set}.
\end{proof}

\underline{\textbf{Proof of Lemma \ref{lemm:x^*=0}:}}
\\\begin{proof}
    Note that reproduction number $\mathcal{R}_0 > 1$ implies that $\beta > \gamma$ holds, which is necessary to ensure that $\It^* = 1 - \frac{\gamma}{\beta}$ exists. Recall, from the characterization under complete rationality in \eqref{eq:sigma=infyFPP} that $$\Delta \ut(\zt^*_{\mathtt{EE(cor)}}) > 0 \implies x^*_{\mathtt{EE(cor)}} = 0, \It^*_{\mathtt{EE(cor)}} = 1 - \frac{\gamma}{\beta}.$$ Substituting these values in \eqref{eq:diff_payoff} results in $$\Delta \ut(\zt^*_{\mathtt{EE(cor)}}) = \frac{\gamma}{\beta} \bigg(\delta_{PS1} + \bigg(\frac{\beta}{\gamma} - 1 \bigg) \delta_{PS0}\bigg) > 0.$$

    \noindent Now assuming that \eqref{eq:delta2} holds, we derive the condition to ensure that SLE $\zt^*_{\mathtt{EE(cor)}}$ is unique. Rewriting \eqref{eq:diff_payoff},
    \begin{align}
    \Delta \ut(\zt^*) &= (1 - \It^*) x^* \delta_{TR1} + \It^* x^* \delta_{TR0} \nonumber
    \\ &\quad +(1 - \It^*) (1 - x^*) \delta_{PS1} + \It^* (1 - x^*) \delta_{PS0}, \nonumber
    \\ &= x^* \big( \It^* (\delta_{TR0} - \delta_{TR1}) + \delta_{TR1} \big) \nonumber
    \\ &\quad+ (1 \! - \! \It^*) (1 \! - \! x^*) (\delta_{PS1} \! - \! \delta_{PS0}) \! + \! (1 - x^*) \delta_{PS0} \nonumber
    \\ &> x^* \big( \It^* (\delta_{TR0} - \delta_{TR1}) + \delta_{TR1} \big) \label{eq:delta3}
    \\ &\quad-\frac{\gamma}{\beta (1 \! - \! x^* \!)} \!\!\times \!(1 \! - \! x^*) \! \! \times \! \! \frac{\beta}{\gamma} \! \times \! \delta_{PS0} \! + \! (1 \! - \! x^* \!)  \delta_{PS0}
    \nonumber
    %\\ &= x^* \big( \It^* (\delta_{TR0} - \delta_{TR1}) + \delta_{TR1} \big) - x^* \delta_{PS0} \nonumber
    \\ &= x^* \big( (1 - \It^*) \delta_{TR1} + \It^* \delta_{TR0} - \delta_{PS0} \big), 
    \label{eq:delta4}
\end{align}
where inequality \eqref{eq:delta3} is obtained by substituting $1 - \It^* = \frac{\gamma}{\beta (1 - x^*)}$ and then by bounding $\delta_{PS1} - \delta_{PS0}$ by the relation in \eqref{eq:delta2}. In order to guarantee that the R.H.S of inequality \eqref{eq:delta4} is non-negative, we derive the condition in the worst-case scenario. When $\delta_{TR0} \leq \delta_{TR1}$ holds, the quantity is non-increasing in $\It^*$ and the minimum occurs at the highest infection level. Since the highest infection level at steady-state is $\It^* = 1 - \frac{\gamma}{\beta}$, therefore, under inequality $\delta_{TR0} \leq \delta_{TR1}$, $\Delta \ut(\zt^*)$ is guaranteed to be non-negative when
\begin{equation}
     \beta (\delta_{TR0} - \delta_{PS0}) + \gamma (\delta_{TR1} - \delta_{TR0}) > 0
     \label{eq:del_tr0_del_tr1_del_ps0}
\end{equation}
holds.

On the other hand, when $\delta_{TR0} > \delta_{TR1}$ holds, the minimum occurs at the lowest infection level $\It^* = 0$. Hence, under inequality $\delta_{TR0} > \delta_{TR1}$, the lower bound of \eqref{eq:delta4} being strictly positive is
\begin{equation}
    \delta_{TR1} - \delta_{PS0} > 0.
    \label{eq:del_tr1_del_ps0}
\end{equation}
In each regime, the inequality corresponding to that condition ($\delta_{TR0} \leq \delta_{TR1}$, or $\delta_{TR0} > \delta_{TR1}$) is always tighter. Thus, the condition is written in compact form as inequality \eqref{eq:min_suff_cond_x^*=0}. Consequently, when \eqref{eq:min_suff_cond_x^*=0} is satisfied, $\Delta \ut(\zt^*)$ is strictly positive for all $x^* \in (0, 1 - \frac{\gamma}{\beta})$ along the endemic branch. At $x^*_{\mathtt{EE(cor)}} = 0$, condition \eqref{eq:delta2} guarantees that $\Delta \ut(0, 1 - \frac{\gamma}{\beta}) > 0$ holds. Thus, $\zt^*_{\mathtt{EE(cor)}} = (0, 1 - \frac{\gamma}{\beta})$ is the unique endemic SLE that exists under the conditions \eqref{eq:delta2} and \eqref{eq:min_suff_cond_x^*=0}.
\end{proof}

\underline{\textbf{Proof of Corollary \ref{thm:x^*=0_stable}:}}
\\\begin{proof}
    From Lemma \ref{lemm:x^*=0}, under the conditions of $\mathcal{R}_0 > 1$ and \eqref{eq:delta2}, the SLE $\zt^*_{\mathtt{EE(cor)}}$ exists. When the payoff difference $\Delta \ut(\zt^*)$ is strictly positive, in the limiting case we observe $$\sigma \rightarrow \infty \implies \frac{\sigma e^{\sigma \Delta \ut(\zt^*_{\mathtt{EE(cor)}})}}{(1 + e^{\sigma \Delta \ut(\zt^*_{\mathtt{EE(cor)}})})^2} \rightarrow 0.$$ Thus, in the limiting case when we substitute $\zt^*_{\mathtt{EE(cor)}} = (0, 1 - \frac{\gamma}{\beta})$ in equations \eqref{eq:idot_i}-\eqref{eq:xdot_x}, we obtain
    \begin{align*}
        &\frac{\partial \dot{\It}(\zt^*_{\mathtt{EE(cor)}})}{\partial \It} = \gamma - \beta < 0,
        \\ &\frac{\partial \dot{\It}(\zt^*_{\mathtt{EE(cor)}})}{\partial x} = -\frac{\gamma}{\beta} (\beta - \gamma) < 0,
        \\ &\frac{\partial \dot{x}(\zt^*_{\mathtt{EE(cor)}})}{\partial \It} = 0,
        \\ &\frac{\partial \dot{x}(\zt^*_{\mathtt{EE(cor)}})}{\partial x} = -1 < 0.
    \end{align*}
The Jacobian $\mathcal{J}(\zt^*_{\mathtt{EE(cor)}})$ is an upper triangular matrix with eigenvalues $\frac{\partial \dot{\It}(\zt^*_{\mathtt{EE(cor)}})}{\partial \It} < 0$ and $\frac{\partial \dot{x}(\zt^*_{\mathtt{EE(cor)}})}{\partial x} < 0$. Thus, we conclude that under complete rationality the unique endemic SLE $\zt^*_{\mathtt{EE(cor)}}$ is LAS.
\end{proof}

\end{document}